\documentclass[a4paper, 11pt]{article}
\usepackage{srcltx,graphicx}
\usepackage{amsmath, amssymb, amsthm}
\usepackage{color, xcolor}
\usepackage{array}
\usepackage{lscape}
\usepackage{algorithm}      
\usepackage[noend]{algpseudocode}
\usepackage{multirow}
\usepackage{psfrag}
\usepackage{pifont}
\usepackage{bbding}
\usepackage[pagebackref]{hyperref}
\usepackage{makeidx}
\usepackage{caption}
\usepackage{subcaption}
\usepackage{afterpage}
\usepackage[section]{placeins}

\definecolor{myred}{rgb}{1,0,0}

\newtheorem{theorem}{Theorem}
\newtheorem{lemma}{Lemma}
\newtheorem{definition}{Definition}

{\theoremstyle{remark} \newtheorem{remark}{Remark}}

\newtheorem{proposition}{Proposition}

\numberwithin{equation}{section}

\graphicspath{{images/}}

\newcommand\showfigures[1]{}

\title{Equilibrium Stability Analysis \\ of Hyperbolic Shallow Water Moment Equations}
\author{
Qian Huang\thanks{Department of Energy and Power Engineering, Tsinghua University, Beijing 100084, China},~
Julian Koellermeier\thanks{Department of Computer Science, KU Leuven, 3001 Leuven, Belgium,
    email: {\tt julian.koellermeier@kuleuven.be}},~
Wen-An Yong\thanks{Department of Mathematical Sciences, Tsinghua University, Beijing 100084, China}
}

\begin{document}
\maketitle


\begin{abstract}
In this paper we analyze the stability of equilibrium manifolds of hyperbolic shallow water moment equations. Shallow water moment equations describe shallow flows for complex velocity profiles which vary in vertical direction and the models can be seen as extensions of the standard shallow water equations. Equilibrium stability is an important property of balance laws that determines the linear stability of solutions in the vicinity of equilibrium manifolds and it is seen as a necessary condition for stable numerical solutions. After an analysis of the hyperbolic structure of the models, we identify three different stability manifolds based on three different limits of the right-hand side friction term, which physically correspond to water-at-rest, constant-velocity, and bottom-at-rest velocity profiles. The stability analysis then shows that the structural stability conditions are fulfilled for the water-at-rest equilibrium and the constant-velocity equilibrium. However, the bottom-at-rest equilibrium can lead to instable modes depending on the velocity profile. Relaxation towards the respective equilibrium manifolds is investigated numerically for different models.
\end{abstract} 
\medskip\noindent
\section{Introduction}
The simulation of free surface flows using shallow water models has led to many successful applications in different scientific fields, such as hydrodynamics \cite{schijf1953theoretical}, snow avalanches \cite{christen2010ramms} and granular flows \cite{craster2009dynamics}. Under the assumption of a constant velocity profile over the water height, the shallow water equations are obtained as an efficient model. However, more complex flows that exhibit varying velocity profiles along the vertical axis from the bottom to the surface cannot be described by this simple model. This is crucial in applications where the bottom friction influences the fluid, e.g. for sediment transport \cite{Garres2020}. Other examples are typical tsunami or dam break scenarios, see \cite{Koellermeier2020}. In \cite{kowalski2018moment} a new model that allows for vertical velocity changes was introduced. The model is called Shallow Water Moment Equations (SWME). Although the model was successfully used for several smooth test cases and moderate dam break scenarios, the lack of hyperbolicity has already been identified by the authors. Hyperbolicity is a mathematical requirement for first-order partial differential equations to be robust against small perturbations of the initial data, a key property of the real-world physical processes \cite{Serre1999}. 
In \cite{Koellermeier2020} it was shown that the model equations (SWME) yield unstable results in the presence of shocks that can be related to the breakdown of hyperbolicity. 

To overcome the lack of hyperbolicity of the SWME model, three new hyperbolic models were recently presented. The first model, called Hyperbolic Shallow Water Moment Equations (HSWME) in \cite{Koellermeier2020} resulted from a straightforward linearization around the constant velocity case , similar to the hyperbolic regularizations in \cite{Cai2013}. The second model, called $\beta$-HSWME, from the same paper includes an additional modification that allows for more control over the eigenvalues of the model. The third model, called Shallow Water Linearized Moment Equations (SWLME) from \cite{Pimentel2020}, was derived from a consistent linearization of only the non-linear terms in the original model. This model lead to an analytical investigation of the steady states and the derivation of an appropriate well-balanced scheme. 

Achieving hyperbolicity is a major step towards applications of the models. This is similar for moment models in rarefied gases originated from the Boltzmann equation, where the so-called Grad model \cite{Grad1949} was not hyperbolic and only the recent hyperbolic regularizations \cite{Cai2013,Koellermeier2014,Fan2016} made the models accessible for applications. 
On the other hand, the stability properties are not solely determined by the hyperbolic transport part but also by the non-negligible right-hand side source term of the models.
For the rarefied gas, the collision source term gives rise to the dissipation behavior of the system which is described by the well-known $H$-theorem \cite{Kremer2010}. It is hence desirable to properly characterize the dissipation property in the shallow water moment models, which contain right-hand side friction terms. 

For this purpose, we resort to a set of structural stability conditions proposed in \cite{Yong1999} for hyperbolic relaxation systems. In short, the conditions impose constrains on the coupling of the source term and the hyperbolic part in the vicinity of the equilibrium. 
It ensures the existence and stability of initial value problems when the relaxation approaches zero \cite{Yong1999}. 
Moreover, the condition is satisfied by many well-developed physical theories \cite{Yong2008}. It was recently shown in \cite{Di2017_2} that the hyperbolic regularized moment models of rarefied gases \cite{Cai2013,Fan2016,Koellermeier2014} fulfill the structural stability condition. In \cite{Huang2020}, it was proven that the Gaussian-type extended quadrature method of moment (EQMOM) for the Boltzmann equation also respects this condition.
On the other hand, \cite{LiuJW2016} reports a counterexample for which the condition is violated so that blow-up solutions are possible for that model.
It is reasonable to believe that the shallow water moment models can only be physically sound if this set of structural stability condition is also fulfilled.

In this paper, we consider the stability properties of hyperbolic shallow water models written in the following form
\begin{equation} \label{eq:1d1orderpde}
\partial_t U + A(U) \partial_x U = S(U),
\end{equation}
in the vicinity of local equilibrium points for which the right-hand side friction term vanishes. Throughout the paper, $U \in \mathbb{R}^{N+2}$ will be the unknown variable, $A(U) \in \mathbb{R}^{(N+2)\times(N+2)}$ the system matrix, and $S(U) \in \mathbb{R}^{N+2}$ the source term. All the models covered in this paper can be written in the form of \eqref{eq:1d1orderpde}. 

In the course of the stability analysis, we will give the first general proofs of hyperbolicity for the HSWME and $\beta$-HSWME models from \cite{Koellermeier2020} by extending the existing proofs, which were only done numerically for models up to size $N=100$ so far. We analyze the equilibrium manifolds on which the right-hand side vanishes, i.e. $S(U)=0$. Apart from the trivial water-at-rest equilibrium, we identify two other equilibrium manifolds in the no-slip and perfect-slip regime, respectively. These equilibrium manifolds model a bottom-at-rest condition and a constant-velocity profile, respectively. The concise analysis of the equilibrium manifolds then allows for the application of the stability conditions first mentioned in \cite{Yong1999}. We prove equilibrium stability for the constant-velocity profile under perfect slip and for the water-at-rest equilibrium. For the bottom-at-rest equilibrium under no-slip, we numerically prove linear instability. Due to the form of the linear instabilities, we discuss that those instabilities only form when the velocity profile allows for a change of sign. This is a complex flow behavior including possible vortexes and back streaming, which leads to a breakdown of the shallow water assumption. We argue that the shallow water moment models are thus stable in the regime of shallow flows. The stability analysis is accompanied by numerical test of dam-break scenarios that highlight the convergence to each equilibrium separately. The test cases are distinguished by different friction values to investigate the behavior of fast relaxation towards the respective equilibrium manifold.  

The rest of this paper is organized as follows: In Section \ref{sec:models}, we recall the shallow water models, among them the three hyperbolic models namely HSWME, $\beta$-HSWME and SWLME. We give novel hyperbolicity proofs for the HSWME and $\beta$-HSWME models for arbitrary $N$. Three equilibrium manifolds for different flow conditions are identified at the end of the section.
Section \ref{sec:analysis} recalls the stability conditions and applies them to the water-at-rest equilibrium and to the constant-velocity equilibrium to show equilibrium stability. In case of the bottom-at-rest equilibrium for the no-slip limit, a numerical example shows the linear instability of all three hyperbolic models. 
Section \ref{sec:simulations} presents numerical simulations of a dam break test case under various friction conditions that show convergence to the three different equilibrium manifolds.
The paper ends with a brief conclusion.

\section{Shallow water moment models}
\label{sec:models}
In this section, we will recall the shallow water moment models, for which we will derive the equilibrium manifolds and subsequently perform a stability analysis.

The standard shallow water equations in one horizontal direction $x$ for water height $h$ and mean velocity $u_m$ using a flat bottom topography are given by
\begin{equation}\label{e:SWE}
    \partial_t
    \begin{pmatrix}
    h\\
    h u_m\\
    \end{pmatrix} +\partial_x
    \begin{pmatrix}
    h u_m\\
    h u_m^2 + \frac{1}{2}g h^2 \\
    \end{pmatrix} =
    -\frac{\nu}{\lambda}
   \begin{pmatrix}
    0\\
    u_m\\
    \end{pmatrix},
\end{equation}
where $\lambda$ is the slip length, and $\nu$ the kinematic viscosity, modeling a Newtonian fluid.

Standard shallow water equations do not allow for the representation of a varying horizontal flow velocity. In other words, the horizontal velocity is constant in vertical direction and only the mean velocity $u_m$ is used. This is consistent with the assumption that the length scale $L$ of the problem is much larger than the water height $h$, i.e. $L\gg h$, leading to only small portions of water flowing in vertical direction. In particular, the flow does not include small features like vortexes. Otherwise the problem requires more complex models like a full solution of the Navier-Stokes equations. 
However, small deviations from constant horizontal velocity profiles often occur in applications, especially together with friction at the bottom, which slows down the flow only at the bottom and leads to a boundary layer close to the bottom of the flow. In the standard shallow water equations, this cannot be represented as the velocity profile is constant. A new model for shallow flows that mitigates this problem was developed in \cite{kowalski2018moment}. The model is based on the following two main ideas:

The first idea is the introduction of a scaled vertical position variable $\zeta(t,x)$, which is defined by
\begin{equation*}
    \zeta(t,x):= \frac{z-h_b(t,x)}{h_s(t,x)-h_b(t,x)}=\frac{z-h_b(t,x)}{h(t,x)},
\end{equation*}
with $h(t,x)=h_s(t,x)-h_b(t,x)$ the water height from the bottom $h_b$ to the surface $h_s$. This transforms the vertical $z$-direction from a physical space to a projected space $\zeta:[0,T]\times \mathbb{R}\rightarrow[0,1]$, see \cite{Koellermeier2020,kowalski2018moment}.

The second idea is a moment expansion of the velocity variable, which is used for expressing more complex velocities, e.g., linear, quadratic or cubic, in the transformed vertical direction. We thus expand $u:[0,T]\times \mathbb{R}\times[0,1]\rightarrow\mathbb{R}$ as
\begin{equation}\label{expansion}
    u(t,x,\zeta)=u_m(t,x)+\sum_{j=1}^{N}\alpha_j(t,x)\phi_j(\zeta).
\end{equation}
Here $u_m:[0,T]\times\mathbb{R}\rightarrow\mathbb{R}$ is the mean velocity and  $\phi_j:[0,1]\rightarrow\mathbb{R}$ are \emph{scaled Legendre polynomials} of degree $j$ defined by
\begin{equation} \label{eq:defphi}
  \phi_j(\zeta) = \frac{1}{j!} \frac{d^j}{d\zeta^j} (\zeta - \zeta^2)^j.
\end{equation}
The first two polynomials read $\phi_1(\zeta) = 1-2\zeta$ and $\phi_2(\zeta) = 1-6\zeta+6\zeta^2$. A basic property of $\phi_j$ is that $\phi_j(0) = 1$. Meanwhile, they form a group of orthogonal basis functions as \cite{kowalski2018moment} 
$$
\int_0^1 \phi_m \phi_n d\zeta = \frac{1}{2n+1} \delta_{mn},
$$
with Kronecker delta $\delta_{mn}$.
$\alpha_j:[0,T]\times\mathbb{R}\rightarrow\mathbb{R}$ with $j\in [1,2,\ldots,N]$ are the corresponding \emph{basis coefficients} at time $t$ and position $x$, also called \emph{moments}. Depending on the values of the coefficients, different horizontal velocity profiles can be described, which leads to an extension compared to the classical shallow water equations \eqref{e:SWE}, where the horizontal velocity is constant, compare \cite{kowalski2018moment}. In the expansion, $N \in \mathbb{N}$ is the order of the velocity expansion and at the same time the maximum degree of the Legendre polynomials. A larger $N$ typically allows for representation of more complex flows, whereas $N=0$ corresponds to the constant velocity profile of the standard shallow water equations \eqref{e:SWE}.

In the shallow water setting, the mean velocity $u_m$ should normally be relatively large in comparison to the deviation part $\sum_{j=1}^{N}\alpha_j(t,x)\phi_j(\zeta)$, see also \cite{kowalski2018moment,Pimentel2020}. Otherwise, small deviations can lead to a change of sign for the velocity profile and back streaming can occur. This can lead to small vortexes in the flow, a phenomenon that is not in agreement with the shallow water regime where the lengths scales are much larger than the water height. In that sense, we only consider velocity distributions that do not include a change of sign, i.e. $u(t,x,\zeta)\geq 0$ without loss of generality. While this seems to be a severe restriction, we note that similar restrictions also hold for other applications, for example, in kinetic theory where the probability density function $f$ is non negative or in rarefied gases where the temperature of the gas is non negative \cite{Koellermeier2017b,Torrilhon2016}.

Equations for the evolution of the basis coefficients are computed by insertion of the expansion into the Navier-Stokes equations, which have been properly transformed to the new $\zeta(t,x)$ variable. Subsequently, the equations are projected onto the Legendre polynomials of degree $i=1,\ldots, N$, which gives one additional equation for each coefficient in the expansion. We refer to \cite{kowalski2018moment} for more details.

The model can be derived in closed form \cite{Pimentel2020} with the help of the precomputed terms $A_{ijk}, B_{ijk}, C_{ij}$ given by 
\begin{equation}
    A_{ijk} = (2i+1) \int_{0}^{1} \phi_i \phi_j \phi_k \,d\zeta,
 \end{equation}
\begin{equation}
    B_{ijk} = (2i+1) \int_{0}^{1} \partial_{\zeta} \phi_i \left( \int_{0}^{\zeta} \phi_j \, d\hat{\zeta} \right) \phi_k \,d\zeta,
\end{equation}
\begin{equation} \label{eq:Cij}
    C_{ij} = \int_{0}^{1} \partial_{\zeta} \phi_i \, \partial_{\zeta} \phi_j \,d\zeta.
\end{equation}

We then write the model as
\begin{equation}\label{SWME_arbitrary}
    \partial_t U + \partial_x F = Q \partial_x U + S,
\end{equation}
with variables $U = \left(h, hu, h\alpha_1, \ldots, h\alpha_N \right)^T \in \mathbb{R}^{N+2}$, the flux Jacobian (also called conservative matrix) $\frac{\partial F}{\partial U}$ given by
\begin{equation*}
    \frac{\partial F}{\partial U} = \begin{bmatrix}
    0 & 1 & 0 & \hdots & 0 \\
    gh -u^2 - \displaystyle\sum_{i=1}^N \frac{\alpha_i}{2i+1} & 2u & \frac{2\alpha_1}{2\cdot 1 +1} & \hdots & \frac{2\alpha_N}{2N+1} \\
    -2u\alpha_1 - \displaystyle\sum_{j,k=1}^N A_{1jk} \alpha_j \alpha_k & 2\alpha_1 & 2 u \delta_{11} + 2 \displaystyle\sum_{k=1}^N A_{11k} \alpha_k & \hdots & 2 u \delta_{1N} + 2 \displaystyle\sum_{k=1}^N A_{1Nk} \alpha_k\\
    \vdots & \vdots & \vdots & \ddots & \vdots\\
    -2u\alpha_N - \displaystyle\sum_{j,k=1}^N A_{Njk} \alpha_j \alpha_k & 2\alpha_N & 2 u \delta_{N1} + 2 \displaystyle\sum_{k=1}^N A_{N1k} \alpha_k & \hdots & 2 u \delta_{NN} + 2 \displaystyle\sum_{k=1}^N A_{NNk} \alpha_k
    \end{bmatrix},
\end{equation*}
and the non-conservative matrix
\begin{equation*}
    Q=\begin{bmatrix}
    0 & 0 & 0 & \hdots & 0 \\
    0 & 0 & 0 & \hdots & 0 \\
    0 & 0 & u \delta_{11} + \displaystyle\sum_{k=1}^N B_{11k} \alpha_k & \hdots & u \delta_{1N} + \displaystyle\sum_{k=1}^N B_{1Nk} \alpha_k\\
    \vdots & \vdots & \vdots & \ddots & \vdots\\
    0 & 0 & u \delta_{N1} + \displaystyle\sum_{k=1}^N B_{N1k} \alpha_k & \hdots & u \delta_{NN} + \displaystyle\sum_{k=1}^N B_{NNk} \alpha_k\\
    \end{bmatrix}.
\end{equation*}
The friction term on the right-hand side $S=(0,S_0,S_1,\dots,S_N)^T \in \mathbb{R}^{N+2}$ is defined in \cite{kowalski2018moment} as $S_0 = 0$ and 
\begin{equation} \label{eq:hswmsrc}
S_i = -\left(2i+1\right) \frac{\nu}{\lambda} \left( u + \sum_{j=1}^N \left( 1 + \frac{\lambda}{h} C_{ij} \right) \alpha_j \right), i=0,\ldots,N.
\end{equation}

The resulting system \eqref{SWME_arbitrary}, called Shallow Water Moment Equations (SWME), can be written in the form of \eqref{eq:1d1orderpde}, where the unknowns are given by $U=[h,hu_m, h\alpha_1,\dots, h\alpha_N]^T \in \mathbb{R}^{N+2}$, the matrix $A(U) \in \mathbb{R}^{(N+2)\times(N+2)}$ models the (conservative and non-conservative) transport part and the right-hand side $S(U) \in \mathbb{R}^{N+2}$ models the friction term. The explicit form of the transport matrix $A= \frac{\partial F}{\partial U} - Q$ can easily be obtained from the explicit terms above. For details on the derivation of the different terms, we refer to \cite{kowalski2018moment}.


We recall the formal definition of hyperbolicity for the 1D first-order PDE of the form \eqref{eq:1d1orderpde}.

\begin{definition}
The system \eqref{eq:1d1orderpde} is called \emph{hyperbolic} if $A(U)$ has $N+2$ linearly independent and real eigenvectors. The system is called \emph{strictly hyperbolic} if $A(U)$ has $N+2$ distinct real eigenvalues.
\end{definition}
Obviously the system is hyperbolic if it is strictly hyperbolic.

It was already noted in \cite{kowalski2018moment}, that the SWME model is not hyperbolic for values $N>1$. Loosing hyperbolicity can lead to instabilities and non-physical values during numerical simulations. In \cite{Koellermeier2020}, the hyperbolicity was studied in more detail and a breakdown of hyperbolicity inducing instable oscillations in time could be found for standard simulations.

\subsection{HSWME}
The SWME model lacks hyperbolicity. The so-called Hyperbolic Shallow Water Moment Equations (HSWME) \cite{Koellermeier2020} overcome this problem by a linearization of the expansion around linear velocity deviations, i.e. the case $N=1$. Effectively, this corresponds to setting $\alpha_i = 0$ for $i=2,\ldots,N$ in the SWME model matrices. This leads to a simplification of the model including modifications in the momentum equation and the higher order moment equations. 

Using $U=[h,hu_m, h\alpha_1,\dots, h\alpha_N]^T \in \mathbb{R}^{N+2}$  
the system is written in non-conservative form as
\begin{equation} \label{eq:hswme}
  \partial_t U + A_H \partial_x U = S(U).
\end{equation}
Here the system matrix $A_H = A_H(U)\in \mathbb{R}^{(N+2)\times(N+2)}$ has the form:
\begin{equation} \label{eq:hswmmat}
    A_H = 
    \begin{bmatrix}
        0 & 1 & &&& \\
        gh-u_m^2-\frac{1}{3}\alpha_1^2 & 2u_m & \frac{2}{3}\alpha_1 &&& \\
        -2u_m\alpha_1 & 2\alpha_1 & u_m & \frac{3}{5}\alpha_1 && \\
        -\frac{2}{3}\alpha_1^2 & 0 & \frac{1}{3}\alpha_1 & u_m & \ddots & \\
        &&& \ddots & \ddots & \frac{N+1}{2N+1}\alpha_1 \\
        &&&& \frac{N-1}{2N-1}\alpha_1 & u_m
    \end{bmatrix},
\end{equation}
which again is the same as using $\alpha_i = 0$ for $i=2,\ldots,N$ in the flux Jacobian and the non-conservative matrix in the original model \eqref{SWME_arbitrary}. Namely, the HSWME model only depends non-linearly on $\alpha_1$ and no longer on the higher coefficients $\alpha_i = 0$ for $i=2,\ldots,N$.

The source term $S(U) = [0,S_0,\dots,S_N]^T \in \mathbb{R}^{N+2}$ of the HSWME model is the same as \eqref{eq:hswmsrc}.

The major advantage of HSWME is that the lack of hyperbolicity of the original SWME is overcome. The previous work \cite{Koellermeier2020} proves that the system matrix $A_H$ is hyperbolic at least for $N \le 150$. 
Indeed, by improving the method in  \cite{Koellermeier2020}, we can further show that

\begin{theorem} \label{thm:hswme}
The HSWME model (\ref{eq:hswme}) of any order $N$ is globally hyperbolic. Moreover, the eigenvalues are
$$
\begin{aligned}
  z_{1,2} &= u_m \pm \sqrt{gh + \alpha_1^2}, \\
  z_{i+2} &= u_m + r_{i,N} \alpha_1, \quad i=1,2,\dots,N,
\end{aligned}
$$
where $r_{i,N} \in \mathbb{R}$ is the $i$-th root of the real polynomial $p_N(z)$ of degree $N$, defined by the recursion $p_k(z) = z p_{k-1}(z) - b_k p_{k-2}(z)$, for
$2 \le k \le N$, $p_1(z) = 1$, $b_k = \frac{(k-1)(k+1)}{(2k-1)(2k+1)}$.
\end{theorem}

\begin{proof}
We first give a sketch on the results in  \cite{Koellermeier2020}. The Theorem 3.2 in \cite{Koellermeier2020} shows that the characteristic polynomial of $A_H$ has the form:
$$
\chi_A(z) = \left( (z-u_m)^2 - gh - \alpha_1^2 \right) \cdot p_N \left(\frac{z-u_m}{\alpha_1} \right),
$$
and $p_N = p_N(z)$ is the characteristic polynomial of the matrix $A_2^{(N)} \in \mathbb{R}^{N\times N}$ defined as
$$
A_2^{(N)} =
  \begin{bmatrix}
    0 & c_2 & & \\
    a_2 & 0 & \ddots & \\
    & \ddots & \ddots & c_N \\
    & & a_N & 0
  \end{bmatrix}.
$$
with entries $a_i = \frac{i-1}{2i-1}$ and $c_i=\frac{i+1}{2i+1}$, for $2 \le i \le N$. 

The first factor of $\chi_A(z)$ gives the two distinct eigenvalues
$$
z_{1,2} = u_m \pm \sqrt{gh + \alpha_1^2}.
$$
From the structure of $A_2^{(N)}$, it is observed in \cite{Koellermeier2020} that $p_N(z)$ follows a three-term recurrence formula:
$$
    p_k(z) = z p_{k-1}(z) - b_k p_{k-2}(z) \text{ for any } 2 \le k \le N,
$$
with $b_i = a_i c_i > 0$.
Then \cite{Koellermeier2020} leaves the job to show formally that $p_N(z)$ has $N$ distinct real roots when $\alpha_1 \ne 0$; Namely, the system is strictly hyperbolic when $\alpha_1 \ne 0$. This was only shown numerically up to $N=150$ in \cite{Koellermeier2020} and will be extended analytically for arbitrary $N \in \mathbb{N}$.

To do this, we prove the following propositions by induction:
\begin{itemize}
\item $p_{k+1}(z)$ has $k+1$ different real roots $z_1^{(k+1)}< z_2^{(k+1)} < \dots < z_{k+1}^{(k+1)}$, and the $k$ roots of $p_k(z)$ lie between them:
  $$
    z_1^{(k+1)} < z_1^{(k)} < z_2^{(k+1)} < \dots < z_{k}^{(k+1)} < z_k^{(k)} < z_{k+1}^{(k+1)}.
  $$
\item If $k$ is even, the signs of the sequence $p_k \left(z_1^{(k+1)}\right), \ p_k \left(z_2^{(k+1)}\right), \dots, \ p_k\left(z_{k+1}^{(k+1)}\right)$ are $(+,-,+,-,\dots,+,-,+)$; otherwise the signs are $(+,-,+,-,\dots,+,-)$.
\end{itemize}

The two properties can be easily verified for $p_0(z)=1$, $p_1(z)=z$, $p_2(z) = z^2-b_2$, and $p_3(z) = z^3 - (b_2+b_3) z$.
Suppose by induction that the properties hold for $p_0(z), p_1(z), \dots, p_{k+1}(z)$, and we look into $p_{k+2}(z)$. Obviously,
$$
p_{k+2} \left(z_i^{(k+1)}\right) = -b_{k+2} p_k\left(z_i^{(k+1)}\right), \text{ for } 1\le i \le k+1.
$$
Now, we focus on the $(k+1)$-term sequence $\{ p_{k+2}\left(z_i^{(k+1)}\right) \}_{1\le i \le k+1}$. According to the inductive assumption, if $k$ is even, the signs of the sequence are $(-,+,-,+,\dots,-,+,-)$; otherwise the signs are $(-,+,-,+,\dots,-,+)$.
Because $p_k$ is a $z$-polynomial of degree $k$, we can conclude in both cases that $p_{k+2}$ has $k+2$ different real roots $z_1^{(k+2)} < z_2^{(k+2)} < \cdots < z_{k+2}^{(k+2)}$, and that the roots of $p_{k+1}$ lie between them:
$$
z_1^{(k+2)} < z_1^{(k+1)} < z_2^{(k+2)} < \cdots < z_{k+1}^{(k+2)} < z_{k+1}^{(k+1)} < z_{k+2}^{(k+2)}.
$$
As for the $(k+2)$-term sequence $\{ p_{k+1}\left(z_i^{(k+2)}\right) \}_{1 \le i \le k+2}$, the signs are $(+,-,+,-,\dots,+,-,+)$ when $k+1$ is even; otherwise the signs are $(+,-,+,-,\dots,+,-)$. 
So this completes the proof of strict hyperbolicity when $\alpha_1 \ne 0$.

When $\alpha_1=0$, the system matrix $A_H$ in Eq.(\ref{eq:hswmmat}) is largely simplified to
$$
\begin{bmatrix}
0 & 1 & \\
gh-u_m^2 & 2u_m & \\
&& u_m I_N
\end{bmatrix},
$$
where $I_N \in \mathbb{R}^{N\times N}$ is the unit matrix.
The system matrix has $N+2$ linearly independent eigenvectors in $\mathbb{R}^{N+2}$: $(1,u+\sqrt{gh},0,\dots,0)^T$, $(1,u-\sqrt{gh},0,\dots,0)^T$, $(0,0,1,\dots,0)^T$, $\dots$, $(0,0,0,\dots,1)^T$. Therefore, the system is still hyperbolic even though it is not strictly hyperbolic.

The form of the eigenvalues can be directly obtained as the roots of the characteristic polynomial of $A_H$.
\end{proof}

We remark that in Section 3 the hyperbolicity is considered as one of the structural stability conditions for this kind of first-order PDEs.

\subsection{$\beta$-HSWME}
The hyperbolic regularization leading to the HSWME model is not the only possibility to change the original SWME and obtain a hyperbolic system of equations. In \cite{Koellermeier2020}, another class of models was presented that allows for some freedom in the eigenvalues of the system. The model is constructed to obtain eigenvalues that can be prescribed by some target polynomial. One example that was analytically derived in \cite{kowalski2018moment} is the so-called $\beta$-HSWME model. The system matrix of the $\beta$-HSWME model only differs from the HSWME model in one entry in the last row:
\begin{equation} \label{eq:b-hswmmat}
    A_{\beta} = 
    \begin{bmatrix}
        0 & 1 & &&& \\
        gh-u_m^2-\frac{1}{3}\alpha_1^2 & 2u_m & \frac{2}{3}\alpha_1 &&& \\
        -2u_m\alpha_1 & 2\alpha_1 & u_m & \frac{3}{5}\alpha_1 && \\
        -\frac{2}{3}\alpha_1^2 & 0 & \frac{1}{3}\alpha_1 & u_m & \ddots & \\
        &&& \ddots & \ddots & \frac{N+1}{2N+1}\alpha_1 \\
        &&&& \frac{2N^2-N-1}{2N^2+N-1}\alpha_1 & u_m
    \end{bmatrix}.
\end{equation}

We now extend the proof of hyperbolicity for the $\beta$-HSWME model from \cite{kowalski2018moment} from $N<100$ to arbitrary $N \in \mathbb{N}$.

\begin{theorem}
The $\beta$-HSWME, with the coefficient matrix in the form of Eq.(\ref{eq:b-hswmmat}), is globally hyperbolic for any order $N$. Moreover, the eigenvalues are
$$
\begin{aligned}
  z_{1,2} &= u_m \pm \sqrt{gh + \alpha_1^2}, \\
  z_{i+2} &= u_m + r_{i,N} \alpha_1, \quad i=1,2,\dots,N,
\end{aligned}
$$
where $r_{i,N}$ is the $i$-th root of the Legendre polynomial of degree $N$.
\end{theorem}

\begin{proof}
Analogously to the proof of Theorem \ref{thm:hswme}, we only need to prove that the $r_{i,N}$'s in the eigenvalues $z_{i+2}$ ($1 \le i \le N$) are the roots of the Legendre polynomial of degree $N$. 
According to \cite{Koellermeier2020} these $r_{i,N}$'s are the eigenvalues of the matrix $A_{\beta,2}^{(N)} \in \mathbb{R}^{N \times N}$ defined as:
$$
A_{\beta,2}^{(N)} =
  \begin{bmatrix}
    0 & c_2 & & \\
    a_2 & 0 & \ddots & \\
    & \ddots & \ddots & c_N \\
    & & a_N' & 0
  \end{bmatrix},
$$
where $a_i = \frac{i-1}{2i-1}$, $c_i=\frac{i+1}{2i+1}$ ($2 \le i \le N$), and the modified last row coefficient reads $a_N' = \beta_{N+1}+a_N = \frac{(N-1)(2N+1)}{(N+1)(2N-1)}$.
Let us denote the characteristic polynomials of $A_{\beta,2}^{(N)}$ and $A_2^{(N)}$ (defined in the proof of Theorem \ref{thm:hswme}) as $p_{\beta,N}=p_{\beta,N}(z)$ and $p_N = p_N(z)$, respectively.
Using the same technique as in the proof of Theorem \ref{thm:hswme}, we see that $p_{\beta,0}=1$, $p_{\beta,1} = z$, $p_{\beta,2}=z^2 - \frac{1}{3}$, and for $N \ge 2$,
$$
p_{\beta,N} = z p_{N-1} - a_N'c_N p_{N-2} = z p_{N-1} - \frac{N-1}{2N-1} p_{N-2}.
$$
We recall the recurrence formula for $p_N$:
$$
p_N = z p_{N-1} - \frac{N^2-1}{4N^2-1} p_{N-2}.
$$

We claim that $p_{\beta,N}$ is exactly the monic Legendre polynomial $\hat{P}_N=\hat{P}_N(z)$ of degree $N$. This is true for $N=0,1,2$. So we only need to verify the recurrence formula for $N \ge 1$:
$$
\hat{P}_{N+1} = z \hat{P}_{N} - \frac{N^2}{4N^2-1} \hat{P}_{N-1}.
$$
This relation comes from the well-known recurrence relation for Legendre polynomials $P_N$: $(N+1)P_{N+1} = (2N+1) z P_{N} - N P_{N-1}$ and that $P_N = \frac{1 \cdot 3 \cdot 5 \cdot \ldots \cdot (2N-1)}{N!} \hat{P}_N$ \cite{Abramowitz1964}.
We prove this equality for $p_{\beta,N}$ by direct calculation:
$$
\begin{aligned}
& \quad p_{\beta,N+1} - z p_{\beta,N} + \frac{N^2}{4N^2-1} p_{\beta,N-1} \\
&= z p_N - \frac{N}{2N+1}p_{N-1} - z \left( zp_{N-1} - \frac{N-1}{2N-1}p_{N-2} \right) + \frac{N^2}{4N^2-1} \left( zp_{N-2} - \frac{N-2}{2N-3} p_{N-3} \right) \\
&= -\frac{N}{2N+1} p_{N-1} + z p_{N-2} \left( -\frac{N^2-1}{4N^2-1} + \frac{N-1}{2N-1} + \frac{N^2}{4N^2-1} \right) - \frac{N^2 (N-2)}{(4N^2-1)(2N-3)} p_{N-3} \\
&= -\frac{N}{2N+1} \left( p_{N-2} - zp_{N-2} + \frac{N(N-2)}{(2N-1)(2N-3)} p_{N-3} \right) = 0,
\end{aligned}
$$
which completes the proof.
\end{proof}

\subsection{SWLME}
In \cite{Pimentel2020} a modified model was introduced based on the linearization of only the nonlinear velocity terms. The model is called Shallow Water Linearized Moment Equations (SWLME). It was originally derived to simplify the nonlinear terms. However, it was shown in \cite{Pimentel2020} that the model allows for an easy analytical evaluation of steady states for the construction of a well-balancing numerical scheme. Moreover, the model is hyperbolic for all $N$, which is important for the stability analysis carried out in this paper. The model's system matrix reads
\begin{equation} \label{eq:n-hswmmat}
    A_{L} = 
    \begin{bmatrix}
        0 & 1 & 0 & \cdots & 0 \\
        gh-u_m^2-\frac{\alpha_1^2}{3} - \cdots - \frac{\alpha_N^2}{2N+1} & 2u_m & \frac{2}{3}\alpha_1 & \cdots & \frac{2}{2N+1}\alpha_N \\
        -2u_m\alpha_1 & 2\alpha_1 & u_m & & \\
        \vdots & \vdots & & \ddots &  \\
        -2u_m \alpha_N & 2\alpha_N && & u_m
    \end{bmatrix}
\end{equation}

According to the hyperbolicity proof in \cite{Pimentel2020} the eigenvalues are given by
\begin{equation*}
    \lambda_{1,2} = u_m \pm \sqrt{gh + \sum_{i=1}^N \frac{3 \alpha_i^2}{2i+1}}, \quad \textrm{ and } \quad \lambda_{i+2} = u_m, \textrm{ for } i=1,\ldots,N.
\end{equation*}
and the eigenvectors $v_i$, $i=1,\ldots, N+2$ are computed as
\begin{equation}
    v_{1,2} = \begin{bmatrix}
    \frac{1}{2 \alpha_n} \\
    \displaystyle\frac{1}{2 \alpha_n} \left( u + \sqrt{gh \pm \sum_{i=1}^N \frac{3 \alpha_i^2}{2i+1}}\right) \\
    \frac{\alpha_1}{\alpha_N} \\
    \vdots \\
    \frac{\alpha_N}{\alpha_N}
    \end{bmatrix}, 
    v_{i+2} = \begin{bmatrix}
    \displaystyle{\frac{6\alpha_{n+1-i}}{(2(n+1-i)+1) -3gh + \sum_{i=1}^N \frac{3 \alpha_i^2}{2i+1}}} \\
    \displaystyle\frac{6\alpha_{n+1-i}u}{(2(n+1-i)+1) -3gh + \sum_{i=1}^N \frac{3 \alpha_i^2}{2i+1}} \\
    \delta_{n+3-i,3} \\
    \vdots \\
    \delta_{n+3-i,N} 
    \end{bmatrix}
\end{equation}
for $i=1,\ldots, N$ and Kronecker delta $\delta_{i,j}$.

For more information on the derivation and properties of the SWLME, we refer to \cite{Pimentel2020}.

\subsection{Equilibrium manifolds}
\label{sec:equilibria}
We have introduced the three different shallow water moment models that fix the loss of hyperbolicity in the original model \cite{kowalski2018moment} by varying the coefficient matrix $A(U)$. 
Note that all the revised models are equipped with the same source term $S(U)$ as given in (\ref{eq:hswmsrc}). This allows to derive the equilibrium states separately from the specific model. This section is devoted to distinguishing the equilibrium states under different parameter settings, which will greatly facilitate our stability analysis in Section 3.

We denote the equilibrium manifold $\mathcal{E} := \{U \in G: S(U) = 0\}$, where $G\subset \mathbb{R}^{N+2}$ is the domain for the unknown variable $U$.

For different scenarios, we distinguish three different equilibrium manifolds, which we cover with the respective subsections below.

\subsubsection{Water-at-rest equilibrium for finite friction coefficients}
Without any assumptions on the friction coefficients $\lambda, \nu$ of the right-hand side source term $S(U)$ in Eq.(\ref{eq:hswmsrc}), we quantify the equilibrium manifold as follows:
\begin{theorem} \label{thm:eqrest}
The equilibrium manifold is 
$$
\mathcal{E} = \{ U \in G: u_m=\alpha_1=\dots=\alpha_N=0 \}.
$$
Namely, the equilibrium state, which is a one-dimensional subspace of $G \subset  \mathbb{R}^{N+2}$, represents the water-at-rest state.
\end{theorem}

To prove the theorem we need to solve $S_i=0$ for $0\le i \le N+1$ based on \eqref{eq:hswmsrc}. This asks for a deeper look into the $C_{ij}$ defined in \eqref{eq:Cij}. We first state the following two Lemmas

\begin{lemma}
For $n\ge 2$,
$$
  \partial_{\zeta} \phi_n = \partial_{\zeta} \phi_{n-2} - (4n-2) \phi_{n-1}.
$$
\end{lemma}

\begin{proof}
Based on the definition of $\phi_n$ in \eqref{eq:defphi}, we have
$$
\begin{aligned}
  \partial_{\zeta} \phi_n &= \frac{1}{n!} \frac{d^{n+1}}{d\zeta^{n+1}} (\zeta-\zeta^2)^n = \frac{1}{(n-1)!} \frac{d^n}{d\zeta^n} \left[ (\zeta-\zeta^2)^{n-1} (1-2\zeta) \right] \\
  &= \frac{1}{(n-1)!} \frac{d^{n-1}}{d\zeta^{n-1}} \left[ (n-1)(\zeta-\zeta^2)^{n-2}(1-2\zeta)^2 - 2(\zeta-\zeta^2)^{n-1} \right] \\
  &= \partial_{\zeta} \phi_{n-2} - (4n-2) \phi_{n-1}.
\end{aligned}
$$
\end{proof}

\begin{lemma}
For $m \le n$, the term $C_{mn}$ in \eqref{eq:Cij} fulfills 
$$
  C_{mn} = C_{nm} = \left \{
  \begin{aligned}
    0 \quad &\text{if $n-m$ is odd,} \\
    2m(m+1) \quad &\text{if $n-m$ is even.}
  \end{aligned}
  \right.
$$
\end{lemma}

\begin{proof}
The symmetry $C_{mn} = C_{nm}$ follows by definition from \eqref{eq:Cij}. Direct calculation gives $C_{11}=4$, $C_{22}=12$, and $C_{12}=0$.

First, we claim for $m\le n$ that $\int_0^1 (\partial_{\zeta} \phi_m) \phi_n d\zeta = 0$. This is because
$$
    \int_0^1 (\partial_{\zeta} \phi_m) \phi_n d\zeta = \int_0^1 \left(\partial_{\zeta} \phi_{m-2}-(4m-2)\phi_{m-1} \right) \phi_n d\zeta = \int_0^1 (\partial_{\zeta} \phi_{m-2}) \phi_n d\zeta.
$$
Therefore, if $m$ is even, the term equals $\int_0^1 (\partial_{\zeta} \phi_0) \phi_n d\zeta = 0$; if $m$ is odd, it equals $\int_0^1 (\partial_{\zeta} \phi_1) \phi_n d\zeta = -2 \int_0^1 \phi_0 \phi_n d\zeta= 0$.

Second, we claim for $m \le n-1$ that $C_{m,n} = C_{m,n-2}$. This is because
$$
    C_{m,n} = \int_0^1 (\partial_{\zeta} \phi_m) \left( \partial_{\zeta} \phi_{n-2} - (4n-2) \phi_{n-1} \right) d\zeta = C_{m,n-2}.
$$

Third, we claim that $C_{m,m+1}=0$. This is because $C_{m,m+1}=C_{m,m-1}=C_{m-1,m}=\dots=C_{12}=0$.
Therefore we have shown that $C_{mn}=0$ if $n-m$ is odd.

Finally, if $n-m$ is even, we have $C_{mn} = C_{mm}$, and
$$
C_{mm}= \int_0^1 \left( \partial_{\zeta} \phi_{m-2} - (4m-2)\phi_{m-1} \right)^2 d\zeta = C_{m-2,m-2} + 4(2m-1).
$$
The final result of $C_{mm}=2m(m+1)$ can be proved by induction on $m$.
\end{proof}

Now we can determine $\mathcal{E}$ by solving $S(U)=0$.
\begin{proof}[Proof of Theorem \ref{thm:eqrest}]
From $S_0 = 0$, we have $u_m + \sum_{j=1}^N \alpha_j = 0$.
Extracting this relation from $S_i=0$ ($i \ge 1$), we obtain $\sum_{j=1}^N C'_{ij} \alpha_j = 0$ with $C'_{ij} = (2i+1) C_{ij}$ ($1\le i \le N$).

We want to show the $N \times N$ matrix $C'$ (with the $ij$th entity being $C'_{ij}$) is invertible.
To this end, for $3 \le n \le N$, we extract the $(n-2)$th column from the $n$th column of $C'$. Because $C'_{m,n}=C'_{m,n-2}$ for $m \le n-1$, this Gaussian transformation yields a lower-triangular matrix, and the diagonal entry becomes $(2m+1)(C_{m,m}-C_{m-2,m-2})>0$.
Thus, $C'$ is invertible and $\alpha_j=0$ for $1\le j \le N$. It then follows that $u_m=0$.

By the ansatz $u(\zeta) = u_m + \sum_j \alpha_j \phi_j(\zeta)$, we see that $u(\zeta) \equiv 0$ at equilibrium.
\end{proof}

Theorem \ref{thm:eqrest} reveals the only possible equilibrium of the shallow water moment models is the water-at-rest state, if no assumptions on the friction coefficients $\lambda, \nu$ are made. 
We shall show in Section 3.2 that this state is stable, indicating that any small perturbation tends to end up at rest.

\subsubsection{Constant-velocity equilibrium for perfect slip model}
In what follows we shall propose two approximations of the source term under different limiting conditions. Let us rearrange $S_i$ in Eq.(\ref{eq:hswmsrc}) as
$$
S_i = \frac{\nu}{\lambda}(2i+1) \left(u_m+\sum_j \alpha_j \right) + \frac{\nu}{h} \sum_j C'_{ij} \alpha_j.
$$
The source term originates from the friction effect of a Newtonian fluid and it contains two terms resulting from integration by parts, see \cite{kowalski2018moment}. 
The relative magnitude of the two terms is decided by the relation between $\lambda$ (the slip length) and $h$.

If $\lambda \gg h$ (namely, $\nu/\lambda \ll \nu/h$), it is possible to neglect the first term of $S_i$. This is the case where $\lambda \to \infty$, a condition representing the Neumann boundary condition with prefect slip \cite{kowalski2018moment}, and the source term reduces to
\begin{equation} \label{eq:srcconstv}
    S_i = -\frac{\nu}{h} \sum_j C'_{ij} \alpha_j, \quad 0 \le i \le N.
\end{equation}
Because the $C'_{ij}$'s ($1\le i,j \le N$) make up an invertible $N \times N$ matrix, the equilibrium state $S(U)=0$ in this case becomes a two-dimensional subspace of $G \subset  \mathbb{R}^{N+2}$:
\begin{equation} \label{eq:eqconstv}
    \mathcal{E} = \{ U \in G: \alpha_1 = \dots = \alpha_N = 0 \}.
\end{equation}
Note that this equilibrium does not impose any condition on the mean velocity $u_m$. It is equivalent to $u(\zeta) = u_m$. Namely, in this perfect-slip limit, the equilibrium state is the constant velocity profile with respect to $\zeta$. We call it the constant-velocity equilibrium.

\subsubsection{Bottom-at-rest equilibrium for no-slip model}
Another possibility to simplify the source term
$$
S_i = \frac{\nu}{\lambda}(2i+1) \left(u_m+\sum_j \alpha_j \right) + \frac{\nu}{h} \sum_j C'_{ij} \alpha_j
$$
is if $\lambda \ll h$ (namely, $\nu/\lambda \gg \nu/h$). A typical scenario is when $\nu$ and $\lambda$ are small parameters of the same order and thus $\kappa=\nu/\lambda = o(1)$. It is then possible to neglect the second term, manifesting the no-slip boundary condition widely used for Newtonian flows. In this case the source term becomes $S(U) = [0,S_0,\dots,S_N]^T$ and for $0 \le i \le N$,
\begin{equation} \label{eq:srcnoslip}
    S_i = - (2i+1) \frac{\nu}{\lambda} \left(u_m+\sum_j \alpha_j \right).
\end{equation}
The equilibrium manifold becomes a hyperplane of $G \subset \mathbb{R}^{N+2}$:
\begin{equation} \label{eq:eqnoslip}
        \mathcal{E} = \{ U \in G: u_m+\sum_j \alpha_j = 0 \}.
\end{equation}
Based on the ansatz this equilibrium is equivalent to $u(\zeta)|_{\zeta=0} = 0$. Namely, in this no-slip limit, the equilibrium state is composed of the velocity profiles which vanish at the bottom $\zeta=0$. We call it the bottom-at-rest equilibrium.
\section{Stability analysis}
\label{sec:analysis}
The equilibrium states $U$, if not dependent on $t$ and $x$, can be viewed as constant solutions to the governing equation (\ref{eq:1d1orderpde}).
It is thus desirable to analyze the stability of the equilibrium states.
Besides the common methods in ODE problems, one must also account for the interaction between the source term $S(U)$ and the convection term $A(U)$, see \cite{Yong1999}.
Based on a structural stability condition, such an analysis is performed in the next section for the various shallow water moment models and equilibria.

\subsection{Structural stability conditions}
In this work, we mainly focus on the structural stability condition proposed in \cite{Yong1999} for first-order PDEs. For simplicity, we herein restate the condition for the 1D equation (\ref{eq:1d1orderpde}).
We denote $S_U(U)$ as the Jacobian of the source term.
The theory specifies whether a state $U$ on the non-empty equilibrium manifold $\mathcal{E}$ is stable and it is stated as below:\\

($\mathbf{I}$): For any $U \in \mathcal{E}$, the Jacobian $S_U(U)$ can be manipulated by an invertible $n\times n$ matrix $P = P(U)$ and an invertible $r\times r$ ($0 < r \leq  n$) matrix $\hat{T}(U)$ such that
$$
P(U) S_U(U) =
\begin{bmatrix}
    0 & 0 \\
    0 & \hat{T}(U) 
\end{bmatrix}
P(U), \quad \forall \ U \in \mathcal{E} .
$$

($\mathbf{II}$): There exists a positive definite symmetrizer $A_0 = A_0(U)$ of the coefficient matrix $A(U)$ such that
$$
A_{0}(U) A(U) = A^{T}(U) A_0(U), \quad \forall \ U \in G.
$$

($\mathbf{III}$): On the equilibrium manifold $\mathcal{E}$, the coefficient matrix and the source term are coupled as
$$
A_0(U)S_U(U) + S_U^T(U)A_0(U) \preceq - P^T(U)
\begin{bmatrix}
    0 & 0 \\
    0 & I_r
\end{bmatrix}
P(U), \quad \forall \ U \in \mathcal{E}.
$$

\begin{remark}
For the 1D system, the condition (II) is equivalent to the requirement of hyperbolicity. If the system is hyperbolic, $A(U)$ has $n$ linearly independent left eigenvectors denoted $r_i$ ($1 \leq  i \leq  n$).
Set $L = [r_1^T, \dots , r_n^T]^T \in \mathbb{R}^{n\times n}$. 
A symmetrizer $A_0$ in the condition (II) can only be of the form $A_0 =  L^T \Lambda L$, with $\Lambda$ an arbitrary positive diagonal matrix.

Therefore, for the three hyperbolic shallow water moment models in Section 2, the condition (II) is already satisfied.
\end{remark}

\begin{remark}
Condition (I) is a common requirement for initial value ODE problems, which can be viewed as spatially homogeneous systems of \eqref{eq:1d1orderpde}.
Condition (III) then provides a criterion on how the spatial convection part should be coupled with the source term in an actual PDE system so that the equilibrium states are stable.
Indeed, one could show by energy estimation techniques that any perturbation from the equilibrium state is bounded.
More detailed discussions and implications can be found in \cite{Yong1999}.
We just point out that many well-developed physical theories are inherently consistent with this condition.
Recently, several moment models originating from kinetic equations have been demonstrated to satisfy this structural stability condition \cite{Di2017_2,Huang2020}.
Therefore, we believe this set of condition can serve as a proper requirement for physically-reasonable moment models.
\end{remark}

In practice, we often work with a sufficient version of condition (III) that is more convenient to handle:

\begin{proposition}[\cite{Yong1999}] \label{prop:suff3}
If $\forall \ U \in \mathcal{E}$, the $n \times n$ matrix $K(M) := P^{-T} A_0 P^{-1} = (\sqrt{\Lambda} L P^{-1})^T (\sqrt{\Lambda} L P^{-1})$ is of block-diagonal form diag$(K_1,K_2)$, in which $K_1$ and $K_2$ are $(n-r)\times (n-r)$ and $r \times r$ matrices, then the system satisfies the structural stability condition (III).
\end{proposition}

In other words, the condition (III) requires the existence of some $P$ (due to condition (I)) and $A_0 = L^T \Lambda L$ (due to condition (II)) such that for the matrix $\sqrt{\Lambda}LP^{-1}$, the first $(n-r)$ columns are orthogonal to the subsequent $r$ columns.

On the contrary, if some equilibrium state $U \in \mathcal{E}$ is unstable, it is impossible to find such a pair of $P$ and $A_0$ to justify condition (III). 
There is indeed a more convenient necessary condition for an unstable equilibrium state:

\begin{proposition} \label{prop:unstable}
An equilibrium state $U_0 \in \mathcal{E}$ is unstable and hence contradicts the structural stability condition (III) if there exists $\xi \in \mathbb{R}$ such that the complex matrix $S_U(U_0) + i \xi A(U_0)$ has an eigenvalue with positive real part.
\end{proposition}

We present here a conceptual indication on why the above proposition implies instability. If $U_0 \in \mathcal{E}$ is stable, then for any solution $U(t,x) = U_0+V(t,x)$ initiated 'close to' $U_0$, the perturbation $V=V(t,x)$ should be bounded with $t \to \infty$.
The governing \eqref{eq:1d1orderpde} can be written as
$$V_t + A(U_0 + V) V_x = S(U_0 + V).$$
Performing Taylor expansion around $U_0$ we obtain
$$V_t + A(U_0)V_x = S_U(U_0)V + o(|V| + |V_x|).$$
A further step of linearization drops the deviation term, yielding:
$$V_t + A(U_0)V_x = S_U(U_0)V. $$
Note this step can only be a good approximation when $V$ is 'sufficiently small'.

Supposing $S_U(U_0) + i\xi A(U_0)$ has an eigenvalue $z$ and an corresponding eigenvector $V_{eig}$, it is straightforward to verify that $V(t,x) = V_{eig} e^{z t - i\xi x}$ is a solution of the above linearized equation.
If the real part of $z$ is positive, $V$ blows up as $t \to \infty$.
Therefore, this condition called \emph{relaxation criterion} breaks the stability of a certain equilibrium state $U_0$.

In what follows we perform the stability analysis of the different shallow water moment models HSWME, $\beta$-HSWME and SWLME from Section \ref{sec:models} under the different parameter settings for the friction term from Section \ref{sec:equilibria}.
As mentioned earlier, all the models are hyperbolic, so the condition (II) is satisfied.

\subsection{Stability of water-at-rest equilibrium for finite friction}
We first assume the physical parameters $\lambda$ and $\nu$ are both finite. We therefore do not perform any simplification of the source term. For the water-at-rest equilibrium from Section \ref{sec:equilibria}, we can then show the following theorem.
\begin{theorem}
All the three shallow water models satisfy the structural stability conditions.
\end{theorem}

\begin{proof}
We subsequently check the different stability conditions.

\textbf{Condition (I)}. As revealed in Theorem \ref{thm:eqrest}, the equilibrium state for all the models is the stationary state $u_m=\alpha_1=\dots=\alpha_N=0$. A direct calculation gives the Jacobian $$S_U(U) = -\frac{\nu}{\lambda h}
\begin{bmatrix}
    0 & \\
    & \hat{S}
\end{bmatrix},
\quad \forall \ U \in \mathcal{E}. $$
The $(N+1)\times(N+1)$ matrix $\hat{S}$ has its $ij$-th element ($0\le i,j \le N$) as $(2i+1)\left( 1+\frac{\lambda}{h}C_{ij} \right)$. (Here we note $C_{0n}=0$.)
We claim $\hat{S}$ is invertible. We only need to extract the 0th column of $\hat{S}$ from all other columns. This gives $\hat{S}_{00}=1$, $\hat{S}_{0n}=0$ for $n\ge 1$ and $\hat{S}_{ij}=(2i+1)\frac{\lambda}{h}C_{ij}$ for $1\le i,j \le N$. The proof of Theorem \ref{thm:eqrest} immediately shows the block $(\hat{S})_{1 \le i,j \le N} = \frac{\lambda}{h}C'$ is invertible, and thus $\hat{S}$ is invertible.
As a result, we can set the required matrix $P$ in the condition (I) as the unit matrix $I_{N+2}$, and $r=N+1$. This verifies the condition (I).

\textbf{Condition (II)}. We explicitly construct the symmetrizer $A_0$ of the three models for later use.
We see from \eqref{eq:hswmmat}, \eqref{eq:b-hswmmat}, \eqref{eq:n-hswmmat} that the coefficient matrices all reduce to the same form at equilibrium:
$$
 \begin{bmatrix}
   0 & 1 & \\
   gh & 0 & \\
   & & 0_{N\times N}
 \end{bmatrix}.
$$
It is not difficult to verify the symmetrizer $A_0 = L^T \Lambda L$ with
$$
  L =
  \begin{bmatrix}
    \sqrt{gh} & 1 & \\
    \sqrt{gh} & -1 & \\
    & & I_{N}
  \end{bmatrix}.
$$

\textbf{Condition (III)}. We follow Proposition \ref{prop:suff3} by setting $\Lambda = I_{N+2}$. Thus $\sqrt{\Lambda} L P^{-1} = L$.
Obviously, the first column of $L$ is orthogonal to all other columns. This verifies the condition (III).
\end{proof}

\subsection{Stability of constant-velocity equilibrium for perfect slip limit}
In this perfect-slip limit, the slip length goes to infinity, $\lambda \to \infty$, and we assume the source term reduces to \eqref{eq:srcconstv}. The equilibrium state is the constant-velocity profile with respect to $\zeta$, as given in \eqref{eq:eqconstv}. We have

\begin{theorem}
All the three shallow water models satisfy the structural stability conditions if the source term reduces to \eqref{eq:srcconstv}.
\end{theorem}

\begin{proof}
\textbf{Condition (I)}. A direct calculation gives the Jacobian at equilibrium as
$$S_U = -\frac{\nu}{h^2} 
\begin{bmatrix}
0_{2\times 2} & \\
& \hat{S}
\end{bmatrix},
\quad \forall \ U \in \mathcal{E}. $$
The $N \times N$ matrix $\hat{S}$ has its $ij$th element as $(2i+1)C_{ij}$ ($1\le i,j \le N$), and is invertible (see proof of Theorem \ref{thm:eqrest}). 
Hence the condition (I) is satisfied by setting $P=I_{N+2}$ and $r=N$.

\textbf{Condition (II)}. We explicitly construct the symmetrizer $A_0$ of the three models for later use. We see from \eqref{eq:hswmmat}, \eqref{eq:b-hswmmat}, \eqref{eq:n-hswmmat}) that the coefficient matrices all reduce to the same form at equilibrium:
$$
 \begin{bmatrix}
   0 & 1 & \\
   gh-u_m^2 & 2u_m & \\
   & & u_m I_{N}
 \end{bmatrix}.
$$
It is not difficult to verify the symmetrizer $A_0 = L^T \Lambda L$ with
$$
  L =
  \begin{bmatrix}
    \sqrt{gh}-u & 1 & \\
    \sqrt{gh}+u & -1 & \\
    & & I_{N}
  \end{bmatrix}.
$$

\textbf{Condition (III)}. We follow Proposition \ref{prop:suff3} by setting $\Lambda = I_{N+2}$. Thus $\sqrt{\Lambda} L P^{-1} = L$.
Obviously, the first and second columns of $L$ are orthogonal to all the subsequent $N$ columns. This verifies the condition (III).
\end{proof}

\subsection{Instability of bottom-at-rest equilibrium for no-slip limit}
In this no-slip limit, the slip length fulfills $\lambda \ll h$. When the viscosity $\nu$ is also a small parameter, we assume the source term reduces to \eqref{eq:srcnoslip}. 
The equilibrium state is the velocity profile vanishing at the bottom $\zeta=0$, as given in \eqref{eq:eqnoslip}.
The Jacobian $S_U(U)$ at equilibrium has rank one.

However, the effort to demonstrate the stability condition (III) based on Proposition \ref{prop:suff3} ended up failing to find a proper positive diagonal $\Lambda \in \mathbb{R}^{(N+2) \times (N+2)}$ for any of the moment models.
It turns out this limit contains unstable equilibrium states for any of the three models, which can be revealed by resorting to Proposition \ref{prop:unstable}. The results are stated as:

\begin{proposition} \label{prop:swm_unstable}
All the three hyperbolic shallow water models contain unstable equilibrium states at least for $N=1,2$ if the source term reduces to \eqref{eq:srcnoslip}.
\end{proposition}

\begin{proof}
We directly verify Proposition \ref{prop:unstable} for each case. Denote for convenience $\kappa = \nu/\lambda$.
\begin{itemize}
    \item \textbf{$\mathbf{N=1}$}. In this case the $\beta$-HSWME and SWLME model coincide with the HSWME; see \eqref{eq:hswmmat}, \eqref{eq:b-hswmmat} \& \eqref{eq:n-hswmmat}. The equilibrium state is $\alpha_1=-u_m$, and the velocity reads $u(\zeta)=2u_m\zeta$ which is a reasonable profile.
    
    We consider a specific equilibrium state $h=\kappa$, $u_m^2 = g \kappa/2$, and $\alpha_1 = -u_m$. Set $\xi = 1/u_m$. A direct calculation shows at this state
    $$
    S_U + i\xi A_H =
    \begin{bmatrix}
    0 & \frac{1}{u_m}i & 0 \\
    \frac{2}{3}u_m i & -1+2i & -1-\frac{2}{3}i \\
    2u_m i & -3-2i & -3+i
    \end{bmatrix}.
    $$
    The characteristic polynomial is $p(z) = -z^3 + (3i-4)z^2 + 11iz + 2i$, and it has a root $0.00757579 + 2.74578i$ with positive real part, which leads to instability.
    
    \item \textbf{$\mathbf{N=2}$, HSWME \& $\beta$-HSWME}. The equilibrium is $u_m+\alpha_1+\alpha_2=0$. We further set $\alpha_1=0$ and the velocity profile at equilibrium becomes $u(\zeta) = 6u_m(\zeta-\zeta^2)$. This is always a physical profile in the region $[0,1]$.
    We note that with this setting the $\beta$-HSWME again coincides with the HSWME; see \eqref{eq:hswmmat} \& \eqref{eq:b-hswmmat}.
    
    We consider a specific equilibrium state $h=\kappa$, $u_m^2 = 2g\kappa$, $\alpha_1=0$ and $\alpha_2=-u_m$. Set $\xi=2/u_m$. A direct calculation shows at this state
    $$
     S_U + i\xi A_H =
    \begin{bmatrix}
    0 & \frac{2}{u_m}i & 0 & 0 \\
    -u_m i & -1+4i & -1 & -1 \\
    0 & -3 & -3+2i & -3 \\
    0 & -5 & -5 & -5+2i
    \end{bmatrix}.
    $$
    The characteristic polynomial is $ z^4 +(9-8i)z^3 - (22+52i) z^2 - (84-24i)z + (8+32i)$, and it has a root $0.00582006 + 0.551959i$ with positive real part, which leads to instability.
    
    \item \textbf{$\mathbf{N=2}$, SWLME}. Similarly, we investigate the equilibrium states with $\alpha_1=0$. In particular, we consider $h=\kappa$, $u_m^2 = g\kappa$, $\alpha_1=0$ and $\alpha_2 = -u_m$. Set $\xi = 5/u_m$. A direct calculation shows at this state
    $$
    S_U + i\xi A_H =
    \begin{bmatrix}
    0 & \frac{5}{u_m}i & 0 & 0 \\
    -u_m i & -1+10i & -1 & -1-2i \\
    0 & -3 & -3+5i & -3 \\
    10u_m i & -5-10i & -5 & -5+5i
    \end{bmatrix}.
    $$
    The characteristic polynomial is $ z^4 +(9-20i)z^3 - (110+150i) z^2 - (555-100i)z - (375-150i)$, and it has a root $0.0238969 + 11.3542i$ with positive real part, which leads to instability.
\end{itemize}
Therefore, we have explicitly identified an unstable equilibrium states of all the three models for $N=1,2$.
\end{proof}

For $N \ge 3$, it is likely that the quadratic equilibrium state $u(\zeta) = 6u_m(\zeta-\zeta^2)$ (that is, $\alpha_2 = -u_m$ and $\alpha_1=\alpha_3=\dots=\alpha_N=0$) is unstable in the no-slip limit.
This can be verified by using Proposition \ref{prop:unstable} for any specified $N$, but a general proof for all $N$ is currently beyond our reach.
Even so, Proposition \ref{prop:swm_unstable} implies the inherent nature of instability in the limit $\lambda \to 0$ (and perhaps $\nu \to 0$).
At least, it disproves the use of \eqref{eq:srcnoslip} as a stable approximation of the source term $S(U)$ in \eqref{eq:hswmsrc}.

\begin{remark}
Following the proof of Proposition \ref{prop:unstable}, we can explicitly derive the 'perturbed' water height and velocity profiles around these equilibrium states that could lead to unbounded solutions as $t \to \infty$. We denote the eigenvalue of $S_U + i\xi A_H$ as $z=z_r + i z_i$, and the corresponding eigenvector $V_e = V_{er} + i V_{ei}$. Then the real solution of the linearized perturbation $V_t + A(U_0) V_x = S_U (U_0)V$ becomes
$$
V(t,x) = v e^{z_r t} \left[ V_{er} \cos (z_i t - \xi x) - V_{ei} \sin (z_i t - \xi x) \right],
$$
with $v>0$ an adjustable parameter. This means the initial value of the variable $\hat{U}$ reads
$$
\hat{U}(0,x) = U_0 + v(V_{er} \cos \xi x + V_{ei} \sin \xi x).
$$

For $N=1$, the specified equilibrium state is $U_0=(\kappa,\kappa u_m, -\kappa u_m)^T$ with $u_m=\sqrt{g \kappa / 2}$ and $\xi=1/u_m$. Denote the $k$th component of $V_e$ to be $V_e^{(k)}=V_{er}^{(k)}+iV_{ei}^{(k)}$. In particular, $V_e^{(1)}=1/u_m$. 
We denote the initial condition of a variable $w$ as $\hat{w}$. Then the initial water height is $\hat{h}(0,x)=\kappa + vu_m^{-1}\cos \xi x$, and we can set $v$ sufficiently small so that $\hat{h}(0,x)>0$ for all $x$. 
For the initial velocity profile, we have
$$
\begin{aligned}
\widehat{hu}(\zeta) &= \widehat{hu_m}(0,x) + \widehat{h\alpha_1}(0,x) \phi_1(\zeta) \\
&= \kappa u_m + v\left( V_{er}^{(2)} \cos \xi x + V_{ei}^{(2)} \sin \xi x \right) + \left[ -\kappa u_m + v \left( V_{er}^{(3)} \cos \xi x + V_{ei}^{(3)} \sin \xi x \right). \right](1-2 \zeta) \\
\end{aligned}
$$
We note that for $\zeta \to 1$, the profile is always positive for sufficiently small $v$. However, for $\zeta \to 0$, the near-wall velocity becomes $ v\left( \left( V_{er}^{(2)}+V_{er}^{(3)} \right) \cos \xi x + \left( V_{ei}^{(2)}+V_{ei}^{(3)} \right) \sin \xi x \right)$, and it turns negative at some positions $x$, regardless of the sign of $z_r$. The velocity of the unstable equilibrium state thus contains a change of sign and does not fulfill the shallow water assumption. 

For $N=2$, the result is similar. For the HSWME, the specified equilibrium state is $U_0 = (\kappa, \kappa u_m, 0, -\kappa u_m)^T$ with $u_m = \sqrt{2g\kappa}$ and $\xi=2/u_m$. The initial velocity profile is
$$
\begin{aligned}
\widehat{hu}(\zeta) =& \widehat{hu_m}(0,x) + \widehat{h\alpha_1}(0,x) \phi_1(\zeta) + \widehat{h \alpha_2}(0,x) \phi_2(\zeta) \\
=& \kappa u_m + v\left( V_{er}^{(2)} \cos \xi x + V_{ei}^{(2)} \sin \xi x \right) + v \left( V_{er}^{(3)} \cos \xi x + V_{ei}^{(3)} \sin \xi x \right) (1-2 \zeta) \\
&+ \left[ -\kappa u_m + v \left( V_{er}^{(4)} \cos \xi x + V_{ei}^{(4)} \sin \xi x \right) \right](1-6 \zeta + 6 \zeta^2).
\end{aligned}
$$
We can show that for sufficiently small $v$, the profile is positive at $\zeta=1/2$ for any $x$. Whereas, for $\zeta \to 0$, the near-wall velocity becomes $v\left( \sum_k V_{er}^{(k)} \cos \xi x + \sum_k V_{ei}^{(k)} \sin \xi x \right)$ which turns negative at some positions $x$. The velocity of the unstable equilibrium state thus again contains a change of sign and does not fulfill the shallow water assumption. 

Therefore, it is revealed that some sign-changing perturbation modes of the no-slip equilibrium velocity profile are unstable in this limit where $\nu$ and $\lambda$ are both small parameters of the same order.
\end{remark}

\begin{remark}
We further remark that this instability seems to be inherited from the original shallow water moment equations in \cite{kowalski2018moment} which could not preserve hyperbolicity. 
We note that when $N=1$, the original system is identical to HSWME (as well as the $\beta$-model and SWLME model). For $N=2$, we can again consider the equilibrium states with $\alpha_1=0$.
In particular, we set $h=\kappa$, $u_m^2 = 35 g \kappa$, $\alpha_1=0$, $\alpha_2=-u_m$, and $\xi = 35/u_m$. Then at this state
$$
A_H = 
\begin{bmatrix}
0 & 1 & 0 & 0 \\
-\frac{41}{35}u_m^2 & 2u_m & 0 & -\frac{2}{5}u_m \\
0 & 0 & 0 & 0 \\
\frac{12}{7}u_m^2 & -2u_m & 0 & \frac{4}{7}u_m
\end{bmatrix},
$$
and the matrix $S_U + i\xi A_H$ has an eigenvalue $0.135419 + 57.5886 i$ with positive real part.
\end{remark}
\section{Simulations}
\label{sec:simulations}

In this section, we will perform simulations for different friction parameters $\nu, \lambda$, such that the three equilibrium manifolds identified in Section \ref{sec:equilibria} are reproduced. All simulations shown here are performed for the models with $N=2$, but the results are qualitatively the same for larger $N$.

The simulations are carried out with a second order scheme path-consistent finite volume scheme based on the implementation \cite{Pimentel2020a}. For more details we refer to the reference.

As initial condition we use a dam-break setup on the domain $[-1,2]$:
\begin{equation*}
    h(0,x) = \left\{\begin{array}{l}
        1.5 \ \ \text{if} \ \ x<0 \\
        1.0 \ \ \text{if} \ \ x>0
        \end{array}\right., u_m(0,x) = 0.25, \alpha_1 = -0.1, \alpha_2 = -0.1,
\end{equation*}
where the initial condition ensures that the velocity profile yields $u(\zeta) > 0$ for all $\zeta \in [0,1]$.

The solution is computed with the second order finite volume scheme taken from \cite{Koellermeier2020}. The scheme uses a path-consistent discretization of the non-conservative products. In addition, a straightforward time splitting between the transport and friction terms is employed. The friction step is then computed using an implicit scheme for the treatment of the possibly stiff right-hand side.  

Three different sets of parameters $\nu, \lambda$ are chosen such that convergence to each respective equilibrium is shown numerically. The specific choices of the parameters are made such that the convergence occurs after reasonable time and the evolution is visible as monitored at times $t=0.1, 0.25, 0.5, 1$. Simulations for larger end times lead to the same results and are not shown here for conciseness.

Figure \ref{fig:water-at-rest} shows the results of the HSWME model for $\nu = 10, \lambda = 1$. The choice of the friction parameters leads to a fast relaxation to the water-at-rest equilibrium $\mathcal{E}_1=\{ U \in G: \alpha_1 = \dots = \alpha_N = 0 \}$. The relaxation can be seen for the solutions at times $t=0.1, 0.25, 0.5, 1$. The variable $EQ_1 = |u_m| + \sum_{i=1}^N |\alpha_i|$ measures the $L_1$-distance from the water-at-rest equilibrium. It can be seen that the flow solution indeed quickly relaxes towards the water-at-rest equilibrium despite the initial shock. The velocity profile in Figure \ref{fig:water-at-rest_uy} clearly converges to the water-at-rest equilibrium. We note that the results are qualitatively the same for the $\beta$-HSWME and SWLME models, which are not shown here.
\begin{figure}[htb!]
		\begin{subfigure}{0.32\textwidth}
		    \centering
			\includegraphics[width=0.99\textwidth]{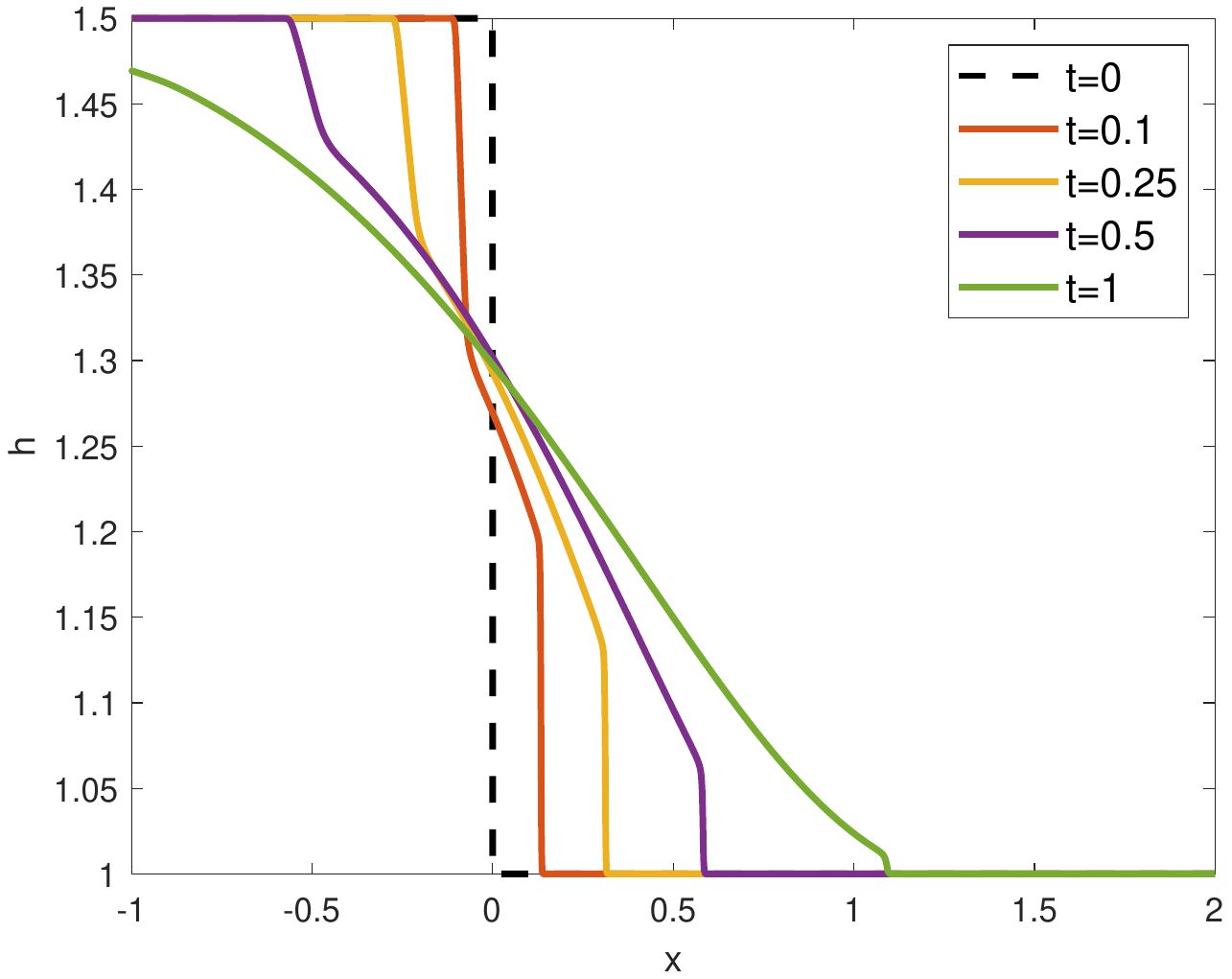}
			\caption{water height $h$}
		    \label{fig:water-at-rest_h}
		\end{subfigure}
		\begin{subfigure}{0.32\textwidth}
		    \centering
			\includegraphics[width=0.99\textwidth]{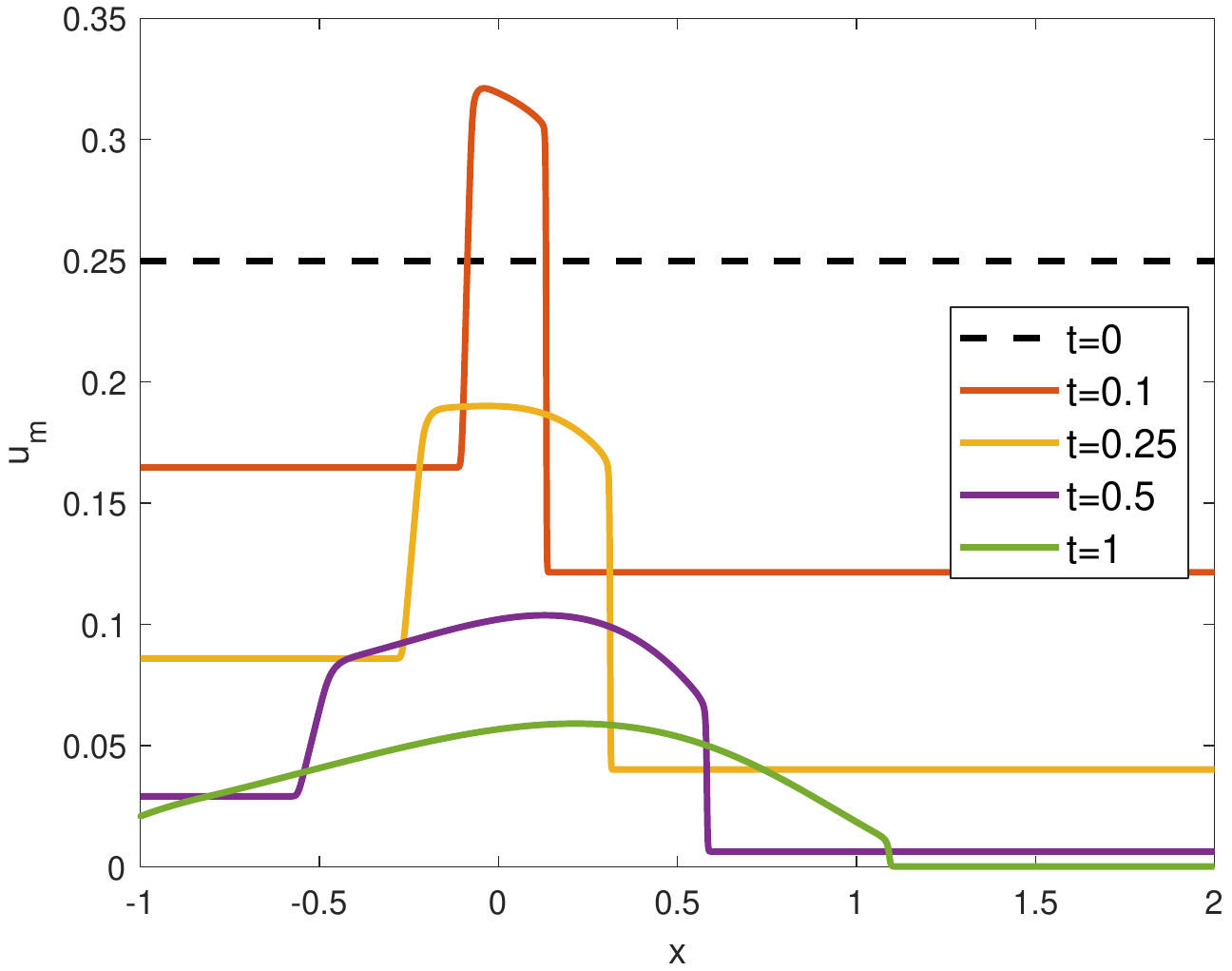}
			\caption{mean velocity $u_m$}
		    \label{fig:water-at-rest_u}
		\end{subfigure}
		\begin{subfigure}{0.32\textwidth}
		    \centering
			\includegraphics[width=0.99\textwidth]{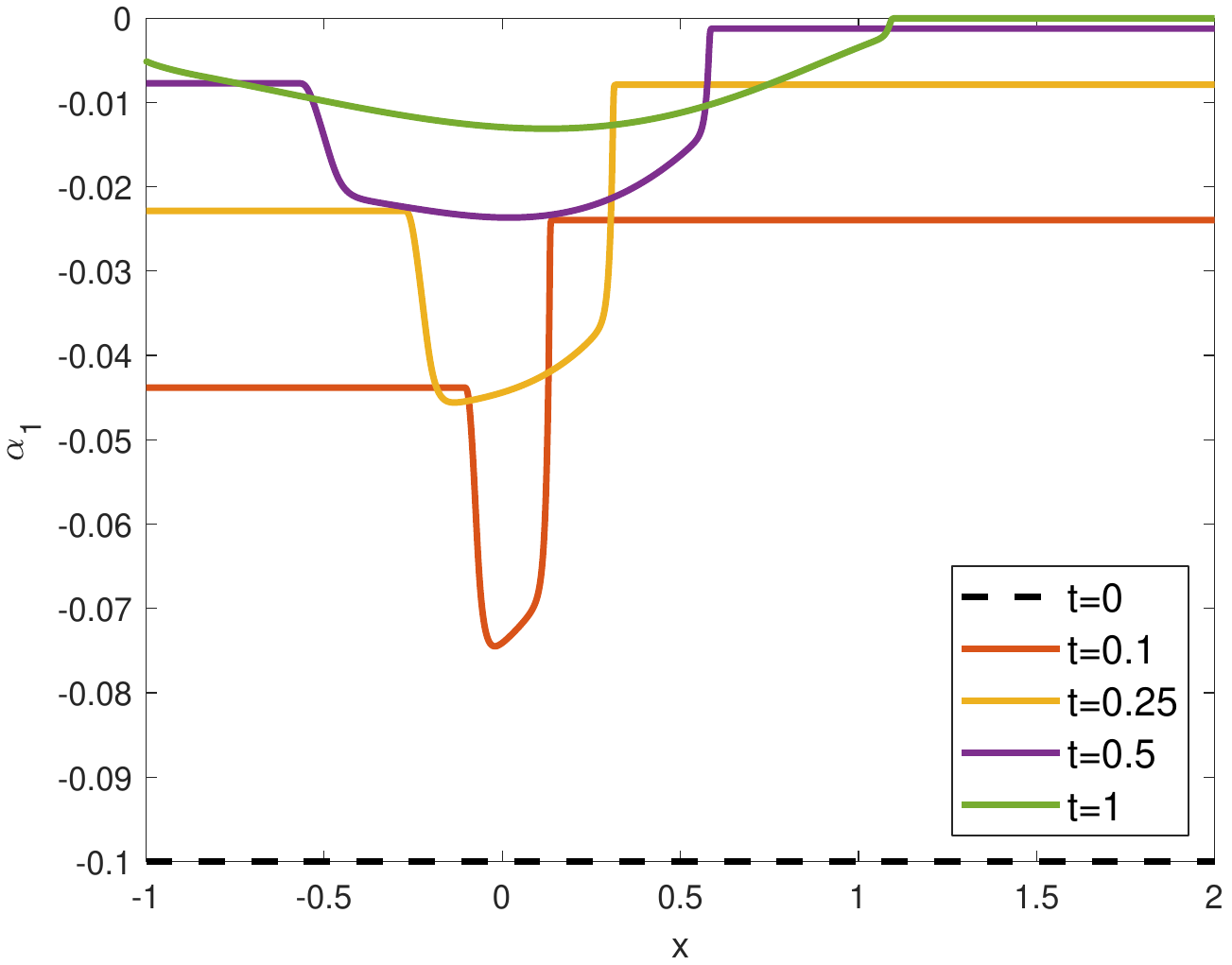}
			\caption{coefficient $\alpha_1$}
		    \label{fig:water-at-rest_alpha1}
		\end{subfigure}
		\center{
		\begin{subfigure}{0.32\textwidth}
		    \centering
			\includegraphics[width=0.99\textwidth]{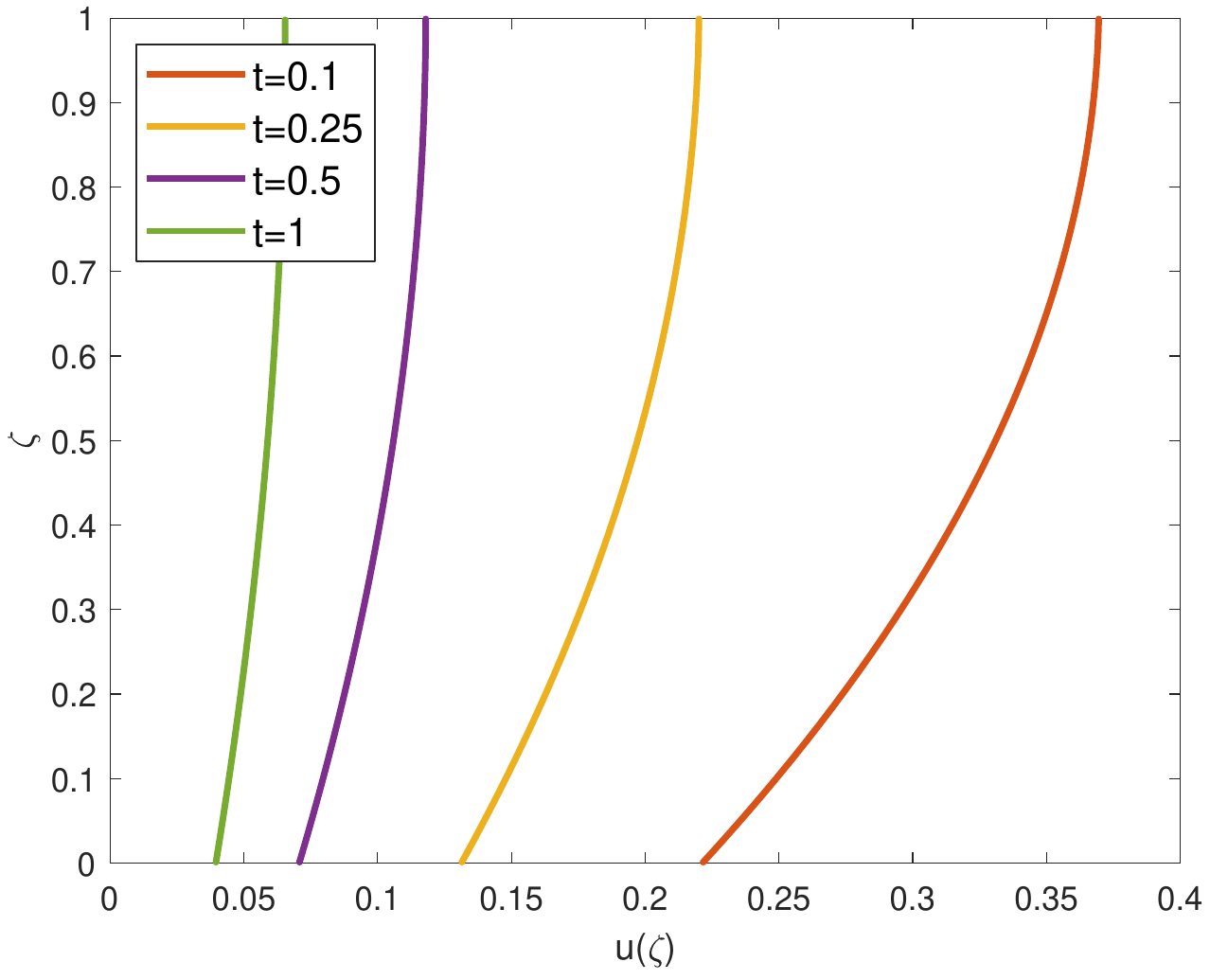}
			\caption{velocity profile $u(\zeta)$, $x=0$}
		    \label{fig:water-at-rest_uy}
		\end{subfigure}
		\begin{subfigure}{0.32\textwidth}
		    \centering
			\includegraphics[width=0.99\textwidth]{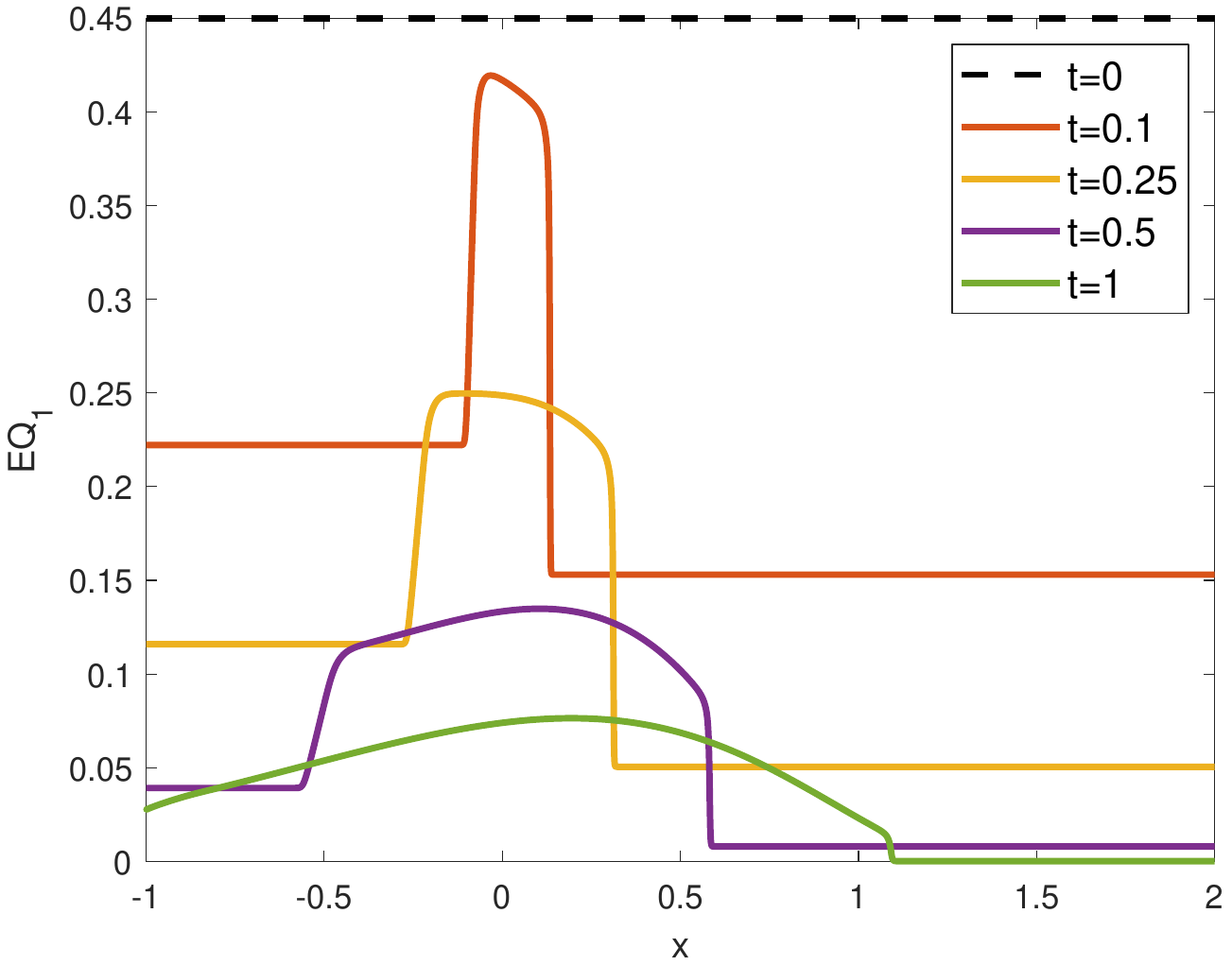}
			\caption{Deviation from equilibrium}
		    \label{fig:water-at-rest_EQ1}
		\end{subfigure}}
        \caption{For $\nu=10, \lambda = 1$, the HSWME model is converging to the water-at-rest equilibrium with time.}
        \label{fig:water-at-rest}
\end{figure}

Figure \ref{fig:constant-velocity} shows the results of the HSWME model for $\nu = 1, \lambda = 10$. The choice of the friction parameters leads to a fast relaxation to the constant-velocity equilibrium $\mathcal{E}_2 = \{ U \in G: \alpha_1 = \dots = \alpha_N = 0 \}$. The relaxation can be seen for the solutions at times $t=0.1, 0.25, 0.5, 1$. The variable $EQ_2 = \sum_{i=1}^N |\alpha_i|$ measures the $L_1$-distance from the constant-velocity equilibrium. It can be seen that the flow solution indeed quickly relaxes towards the constant-velocity equilibrium. The velocity profile in Figure \ref{fig:constant-velocity_uy} clearly converges to the constant-velocity equilibrium. Again, we note that the results are qualitatively the same for the $\beta$-HSWME and SWLME models, which are not shown here.
\begin{figure}[htb!]
		\begin{subfigure}{0.32\textwidth}
		    \centering
			\includegraphics[width=0.99\textwidth]{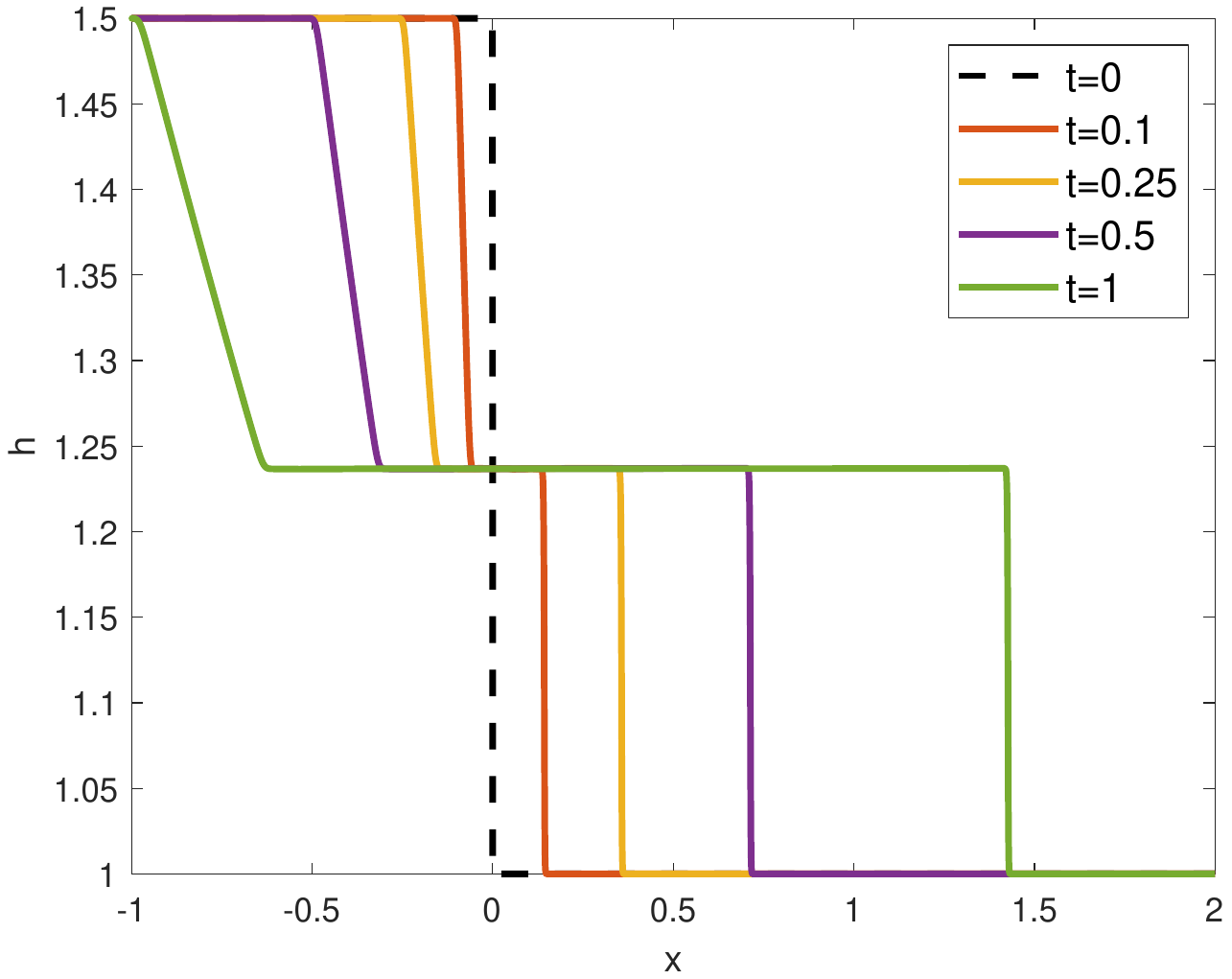}
			\caption{water height $h$}
		    \label{fig:constant-velocity_h}
		\end{subfigure}
		\begin{subfigure}{0.32\textwidth}
		    \centering
			\includegraphics[width=0.99\textwidth]{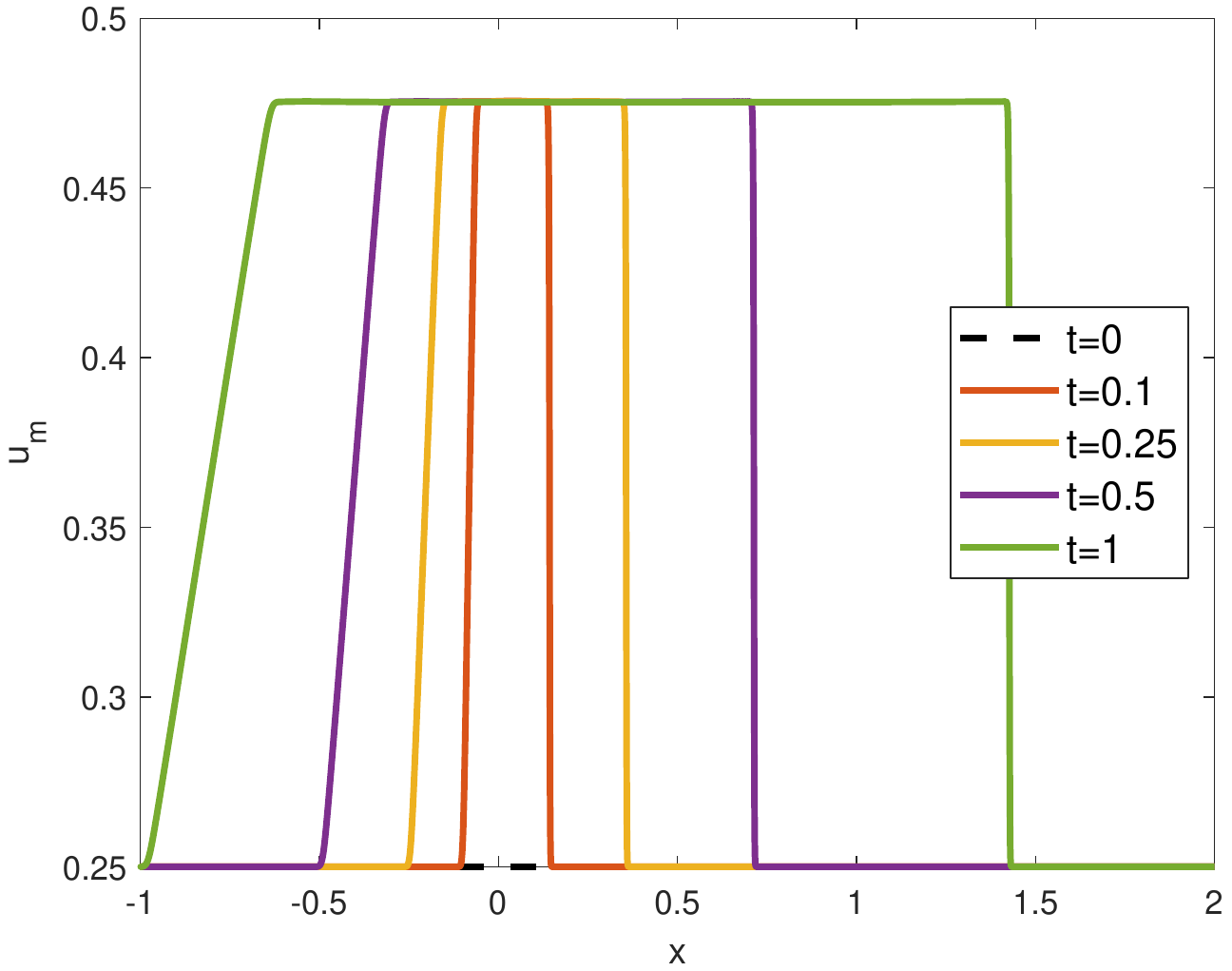}
			\caption{mean velocity $u_m$}
		    \label{fig:constant-velocity_u}
		\end{subfigure}
		\begin{subfigure}{0.32\textwidth}
		    \centering
			\includegraphics[width=0.99\textwidth]{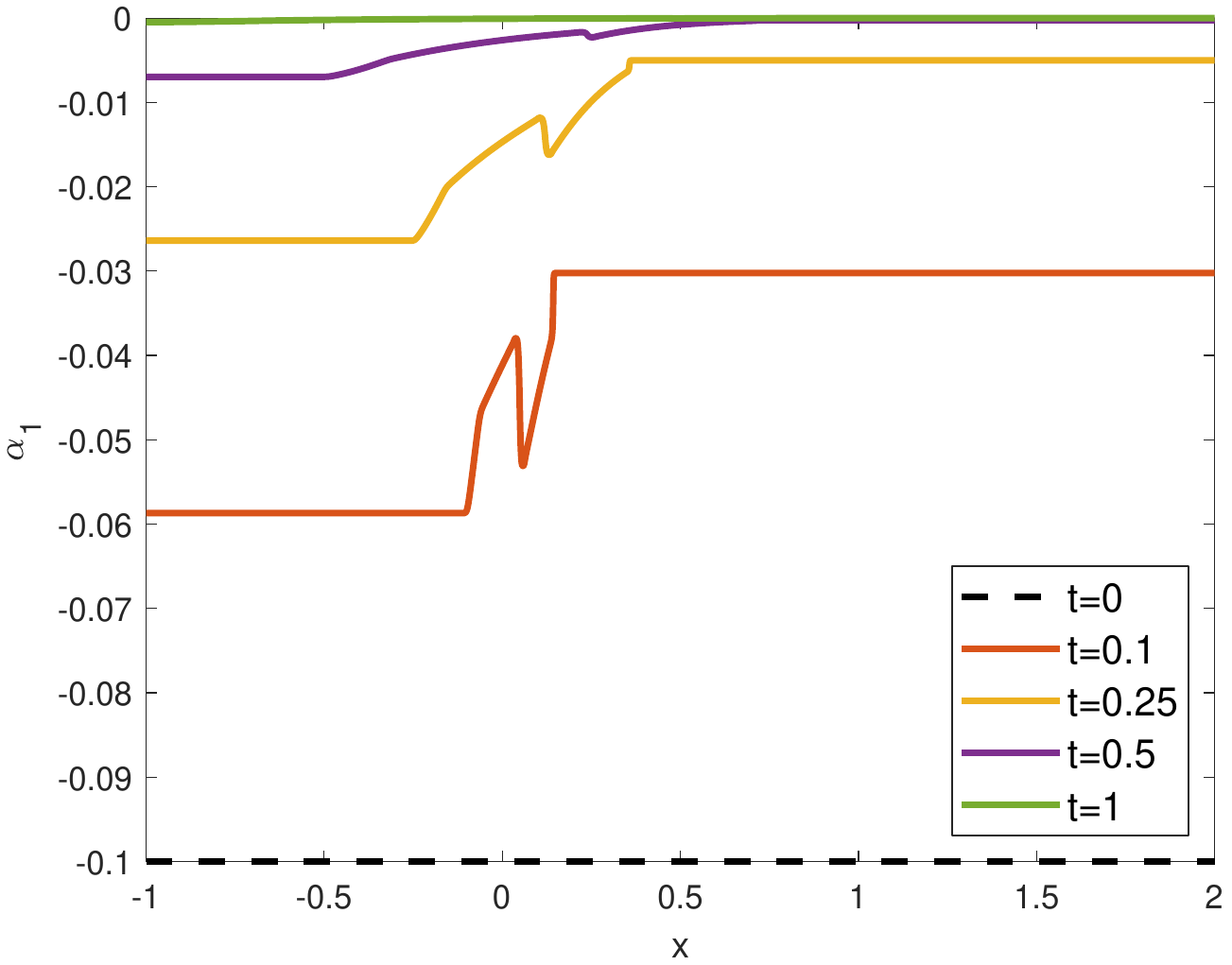}
			\caption{coefficient $\alpha_1$}
		    \label{fig:constant-velocity_alpha1}
		\end{subfigure}
		\center{
		\begin{subfigure}{0.32\textwidth}
		    \centering
			\includegraphics[width=0.99\textwidth]{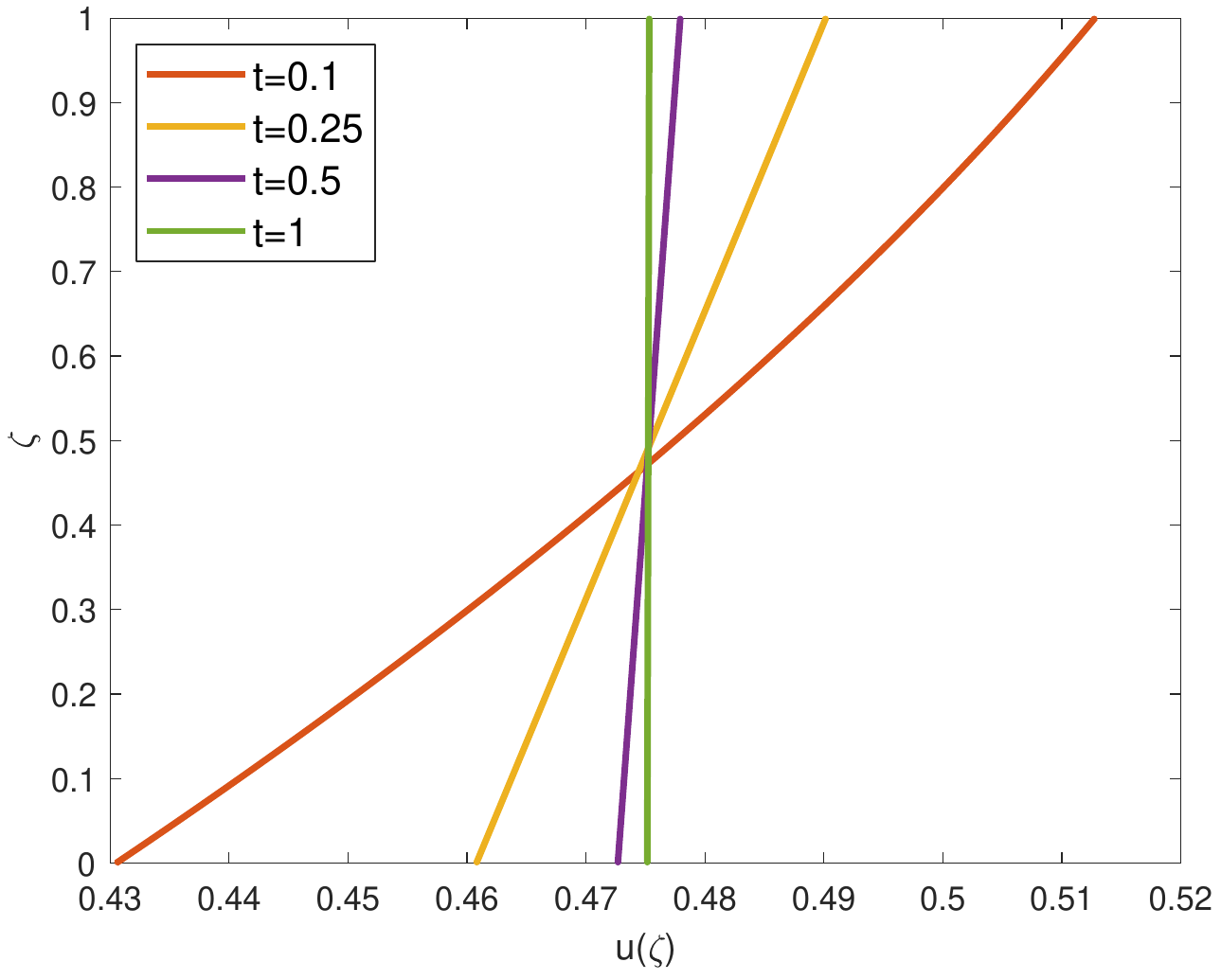}
			\caption{velocity profile $u(\zeta)$, $x=0$}
		    \label{fig:constant-velocity_uy}
		\end{subfigure}
		\begin{subfigure}{0.32\textwidth}
		    \centering
			\includegraphics[width=0.99\textwidth]{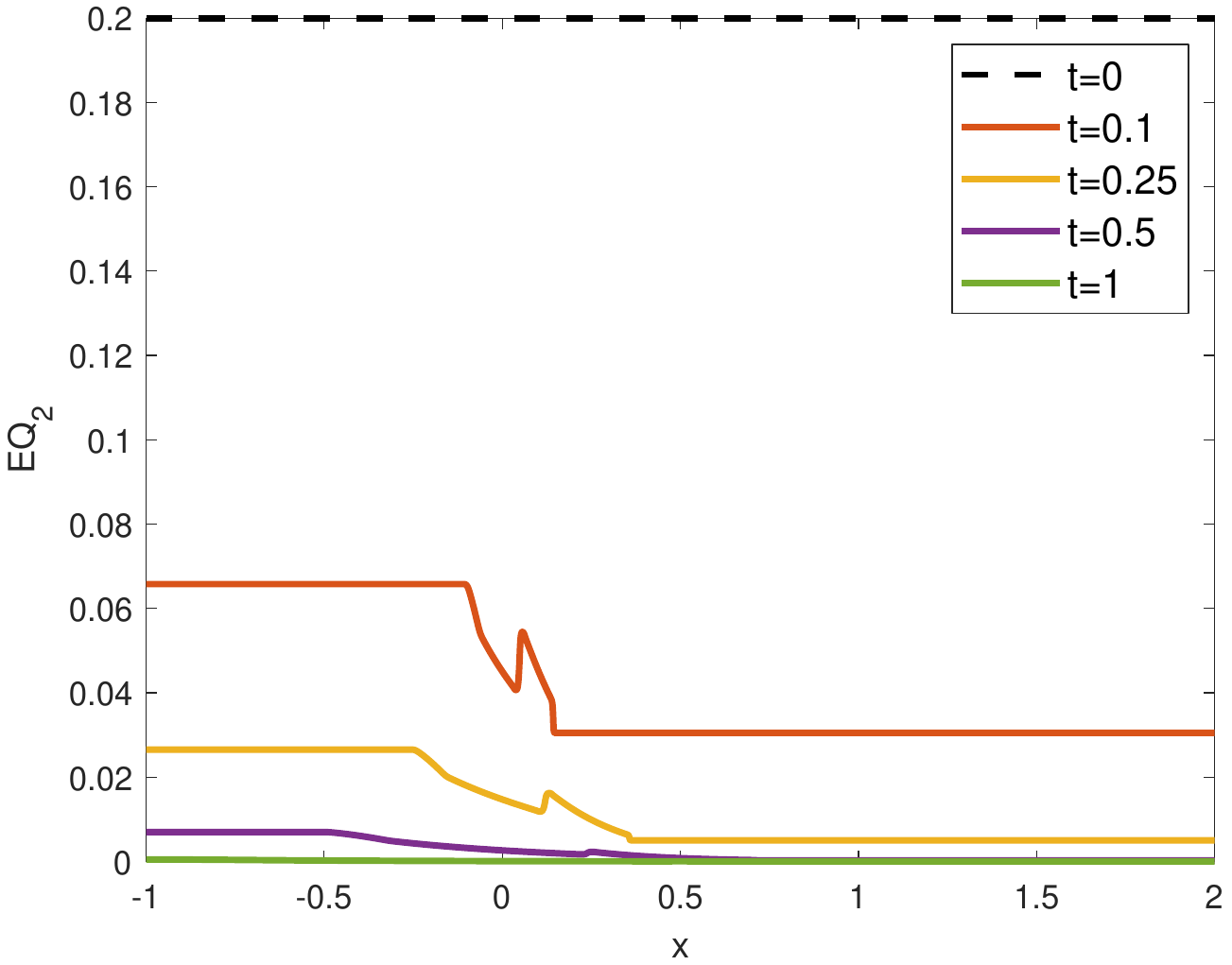}
			\caption{Deviation from equilibrium}
		    \label{fig:constant-velocity_EQ2}
		\end{subfigure}}
        \caption{For $\nu=1, \lambda = 10$, the HSWME model is converging to the constant-velocity equilibrium with time.}
        \label{fig:constant-velocity}
\end{figure}

Figure \ref{fig:bottom-at-rest} shows the results of the HSWME model for $\nu = 0.001, \lambda = 0.001$. The choice of the friction parameters leads to a fast relaxation to the bottom-at-rest equilibrium $\mathcal{E}_3 = \{ U \in G: u_m+\sum_j \alpha_j = 0 \}$. The relaxation can be seen for the solutions at times $t=0.1, 0.25, 0.5, 1$. The variable $EQ_3 = |u_m + \sum_{i=1}^N \alpha_i|$ measures the $L_1$-distance from the bottom-at-rest equilibrium, i.e. the deviation from zero velocity at the bottom. It can be seen that the flow solution indeed quickly relaxes towards the bottom-at-rest equilibrium. Note that the bottom-at-rest equilibrium here is different from the water-at-rest equilibrium in Figure \ref{fig:water-at-rest}, where also the mean velocity converges to zero. The velocity profile in Figure \ref{fig:bottom-at-rest_uy} clearly converges to the bottom-at-rest equilibrium. As shown in Section \ref{sec:analysis}, this equilibrium is linearly unstable. In this full non-linear simulation we thus may or may not see instabilities emerging. There is indeed a small instability building up in Figure \ref{fig:bottom-at-rest_EQ3}, which is best visible for time $t=1$ around $x=0.5$. The instability clearly grows in time as predicted. This can also be seen in tests with larger end times, which are not shown here for conciseness. This test case numerically shows that the bottom-at-rest equilibrium leads to a non-linear instability for the HSWME model. We note that the results are qualitatively the same for the $\beta$-HSWME model, which is not shown here. However, the SWLME yields a slightly different non-linear stability behavior, as seen in the next figure.
\begin{figure}[htb!]
		\begin{subfigure}{0.32\textwidth}
		    \centering
			\includegraphics[width=0.99\textwidth]{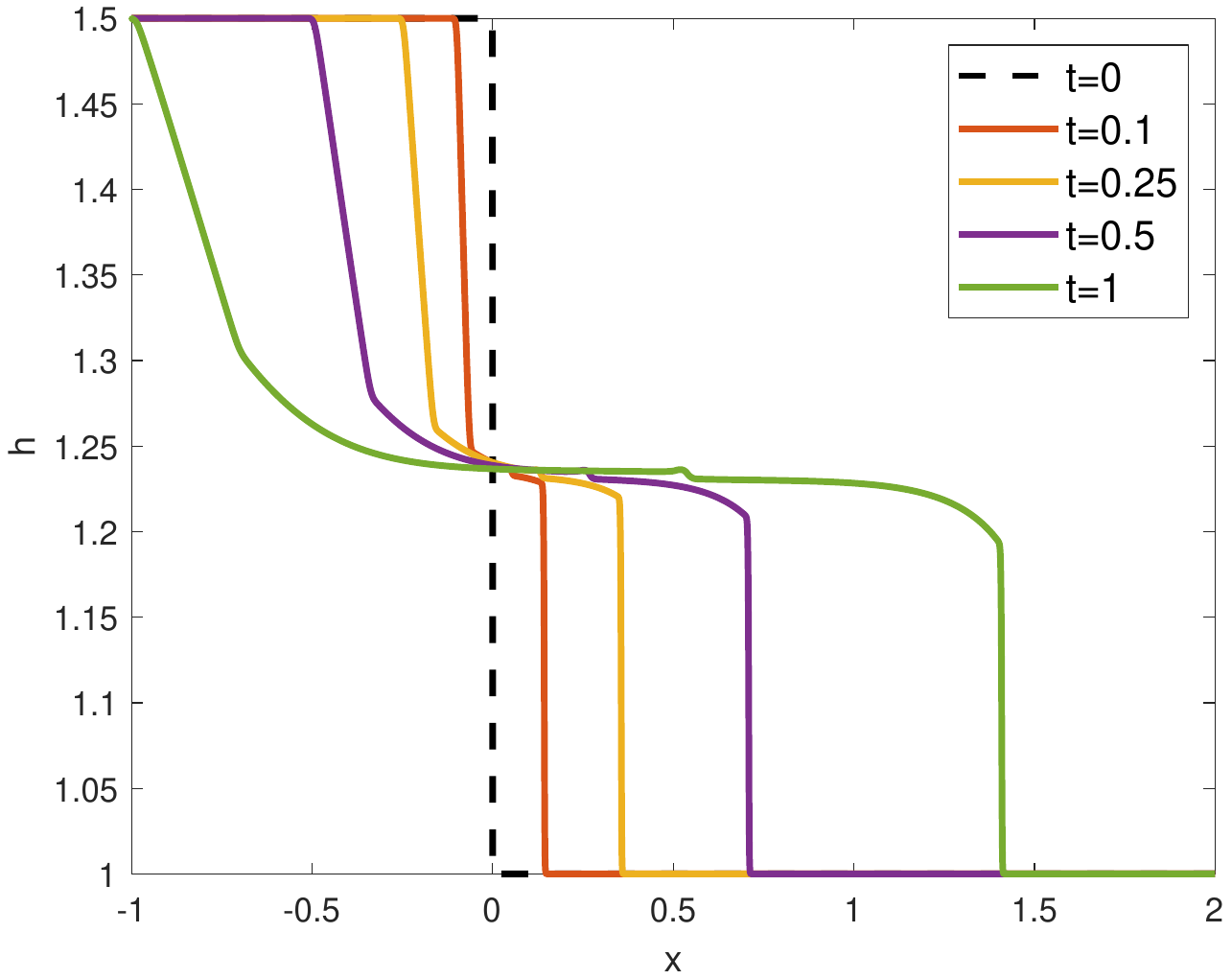}
			\caption{water height $h$}
		    \label{fig:bottom-at-rest_h}
		\end{subfigure}
		\begin{subfigure}{0.32\textwidth}
		    \centering
			\includegraphics[width=0.99\textwidth]{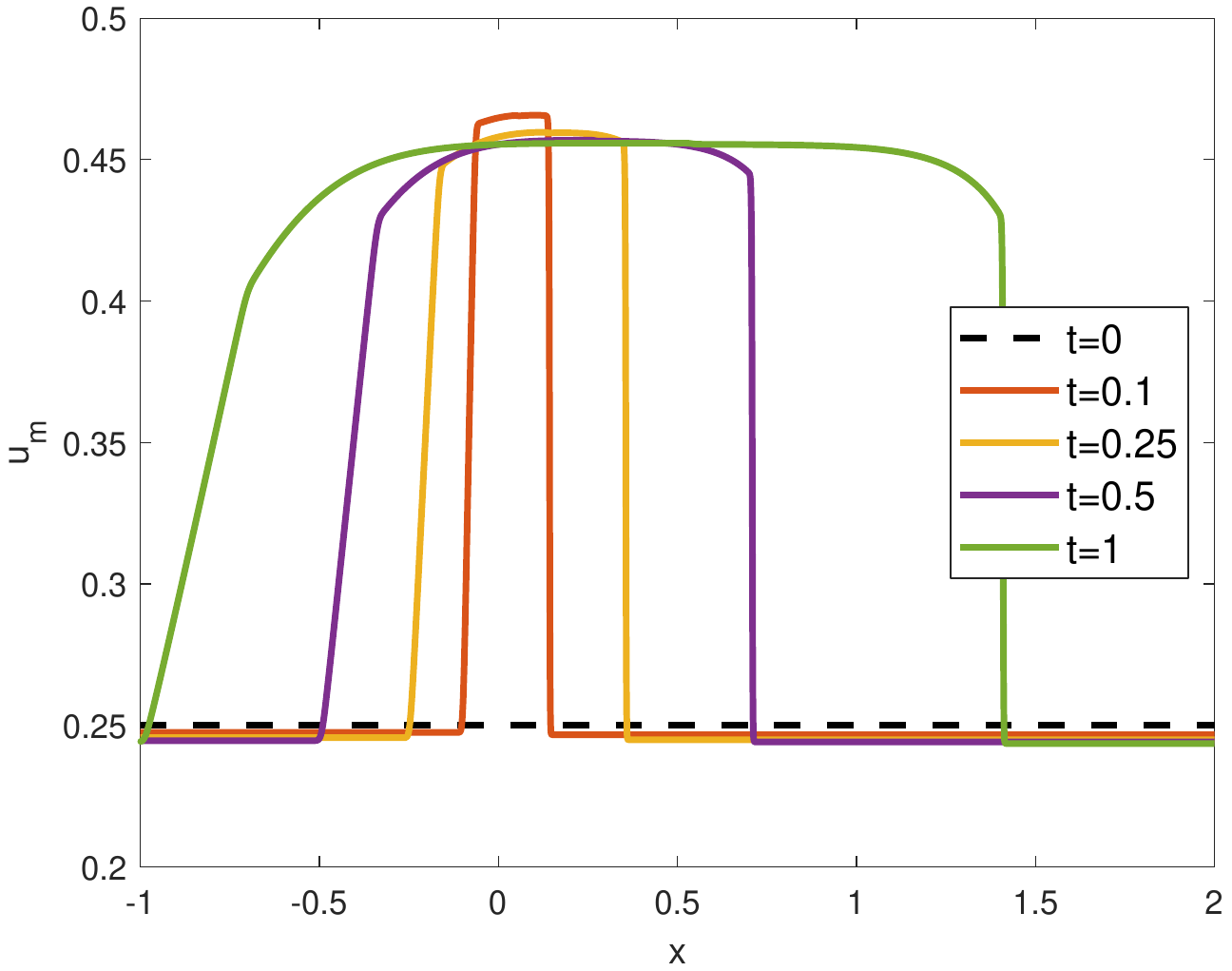}
			\caption{mean velocity $u_m$}
		    \label{fig:bottom-at-rest_u}
		\end{subfigure}
		\begin{subfigure}{0.32\textwidth}
		    \centering
			\includegraphics[width=0.99\textwidth]{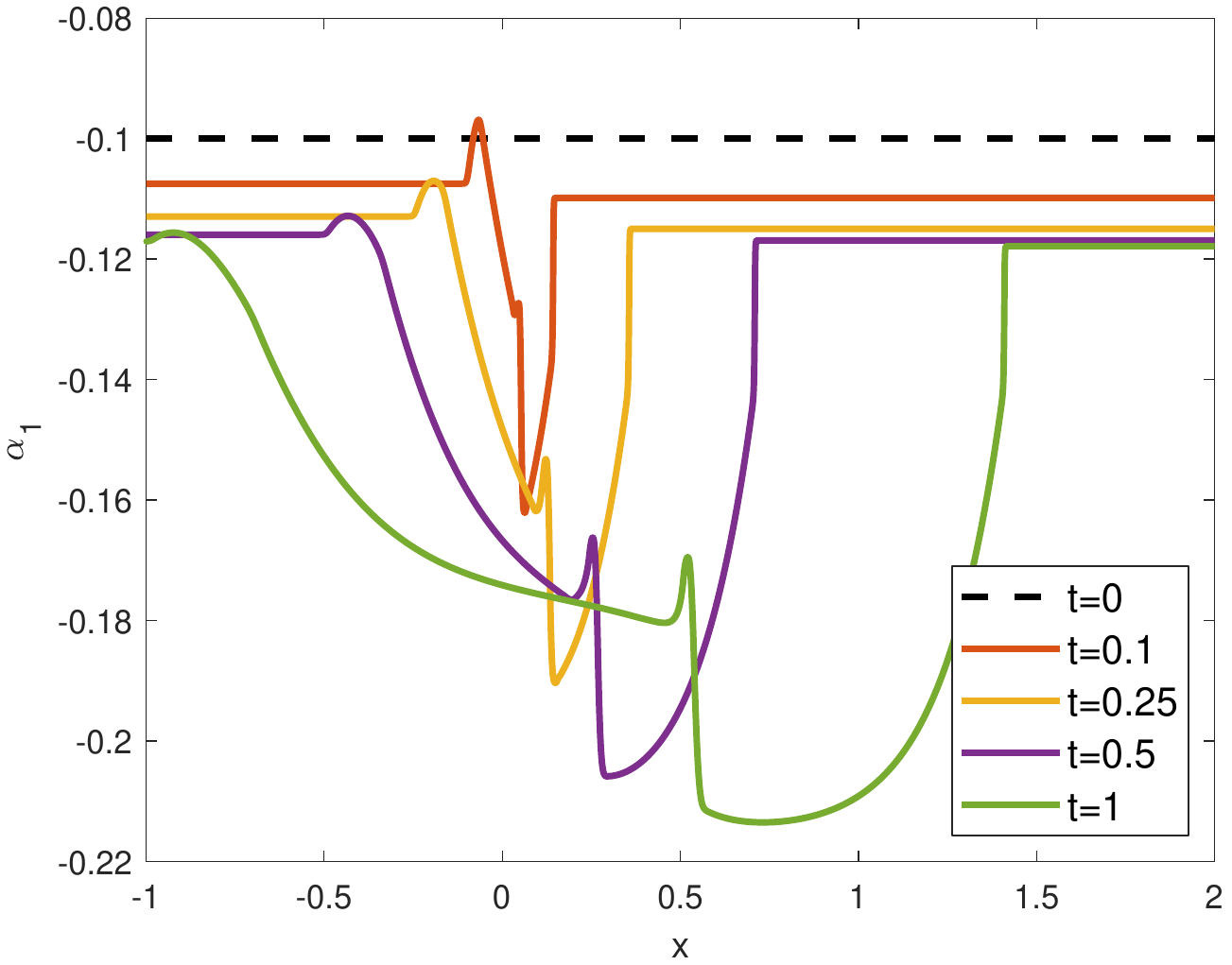}
			\caption{coefficient $\alpha_1$}
		    \label{fig:bottom-at-rest_alpha1}
		\end{subfigure}
		\center{
		\begin{subfigure}{0.32\textwidth}
		    \centering
			\includegraphics[width=0.99\textwidth]{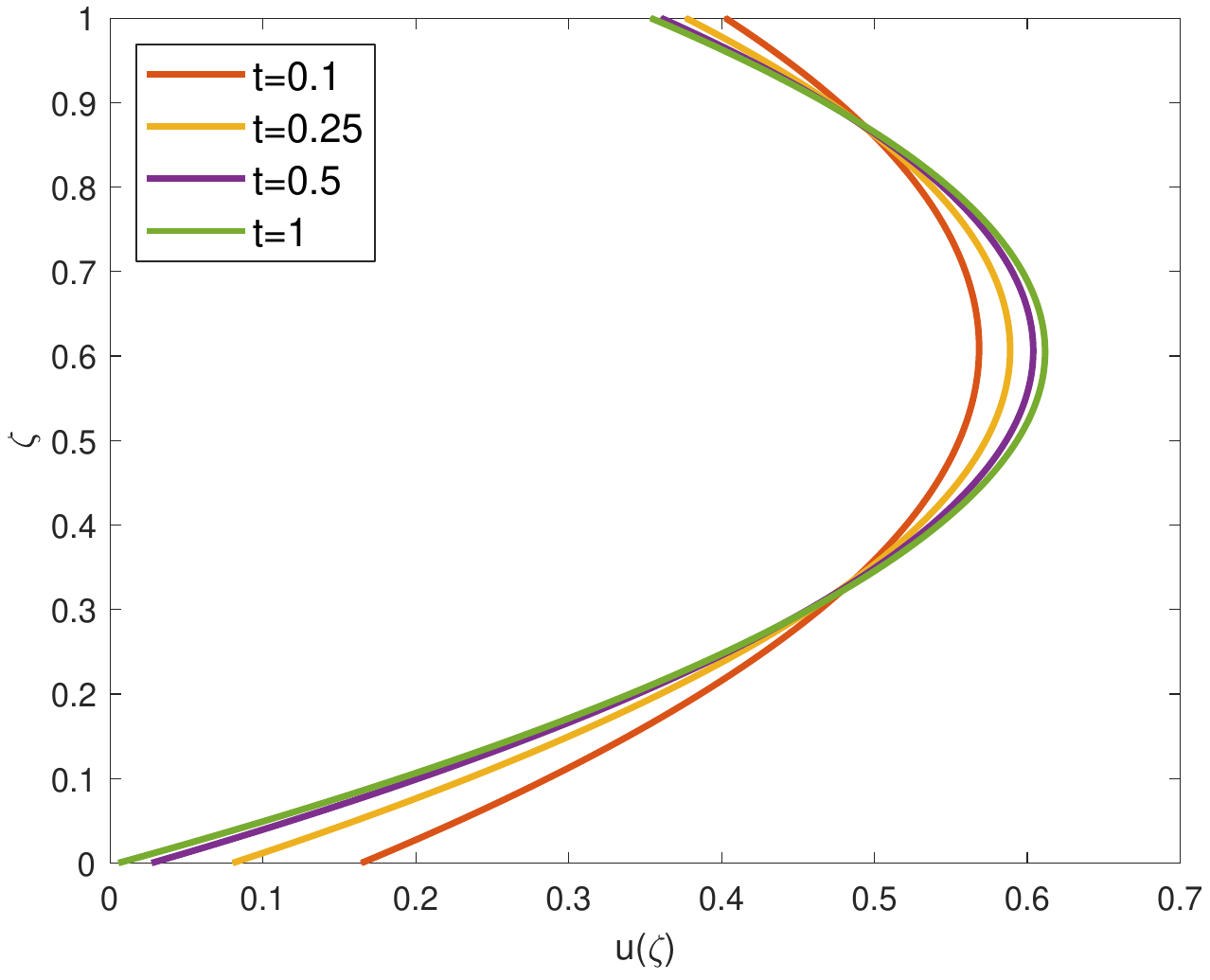}
			\caption{velocity profile $u(\zeta)$, $x=0$}
		    \label{fig:bottom-at-rest_uy}
		\end{subfigure}
		\begin{subfigure}{0.32\textwidth}
		    \centering
			\includegraphics[width=0.99\textwidth]{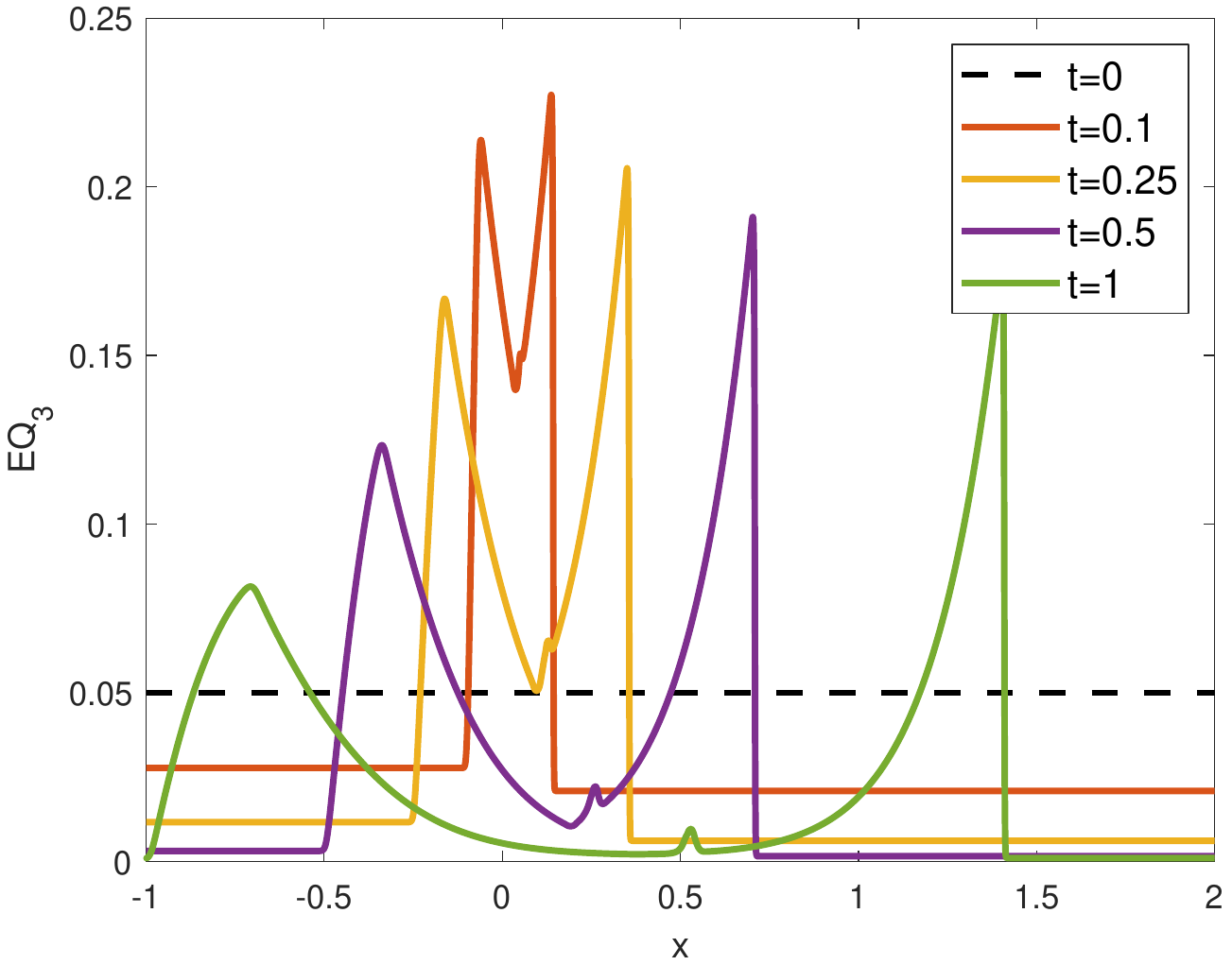}
			\caption{Deviation from equilibrium}
		    \label{fig:bottom-at-rest_EQ3}
		\end{subfigure}}
        \caption{For $\nu=10^{-3}, \lambda = 10^{-3}$, the HSWME model is converging to the bottom-at-rest equilibrium with time. A small instability is forming and propagating downstream.}
        \label{fig:bottom-at-rest}
\end{figure}

Figure \ref{fig:new_bottom-at-rest} shows the results of the SWLME model for $\nu = 0.0011, \lambda = 0.001$. The results are the same as for the HSWME model in Figure \ref{fig:bottom-at-rest}, despite the fact that the instability is not visible in the results. The missing instability is in agreement with the propagation speeds and the eigenstructure of the system. The wave that is causing the instability in the HSWME model test case above is removed in the SWLME model. This test case numerically shows that the bottom-at-rest equilibrium does not lead to a non-linear instability for the SWLME model. We note that this instability is also not forming for larger end times. It seems that the SWLME model has better non-linear stability properties than the HSWME and $\beta$-HSWME models.
\begin{figure}[htb!]
		\begin{subfigure}{0.32\textwidth}
		    \centering
			\includegraphics[width=0.99\textwidth]{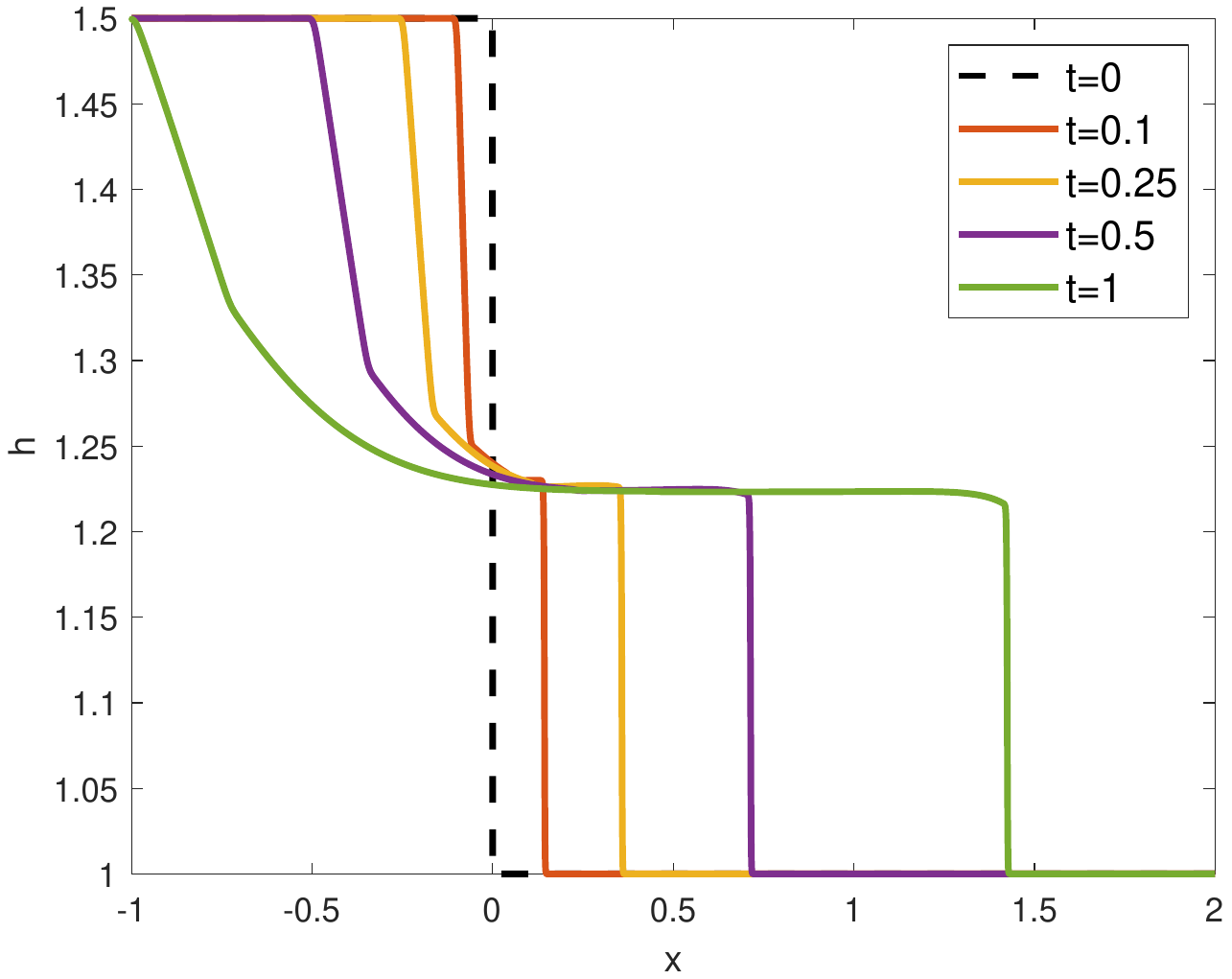}
			\caption{water height $h$}
		    \label{fig:new_bottom-at-rest_h}
		\end{subfigure}
		\begin{subfigure}{0.32\textwidth}
		    \centering
			\includegraphics[width=0.99\textwidth]{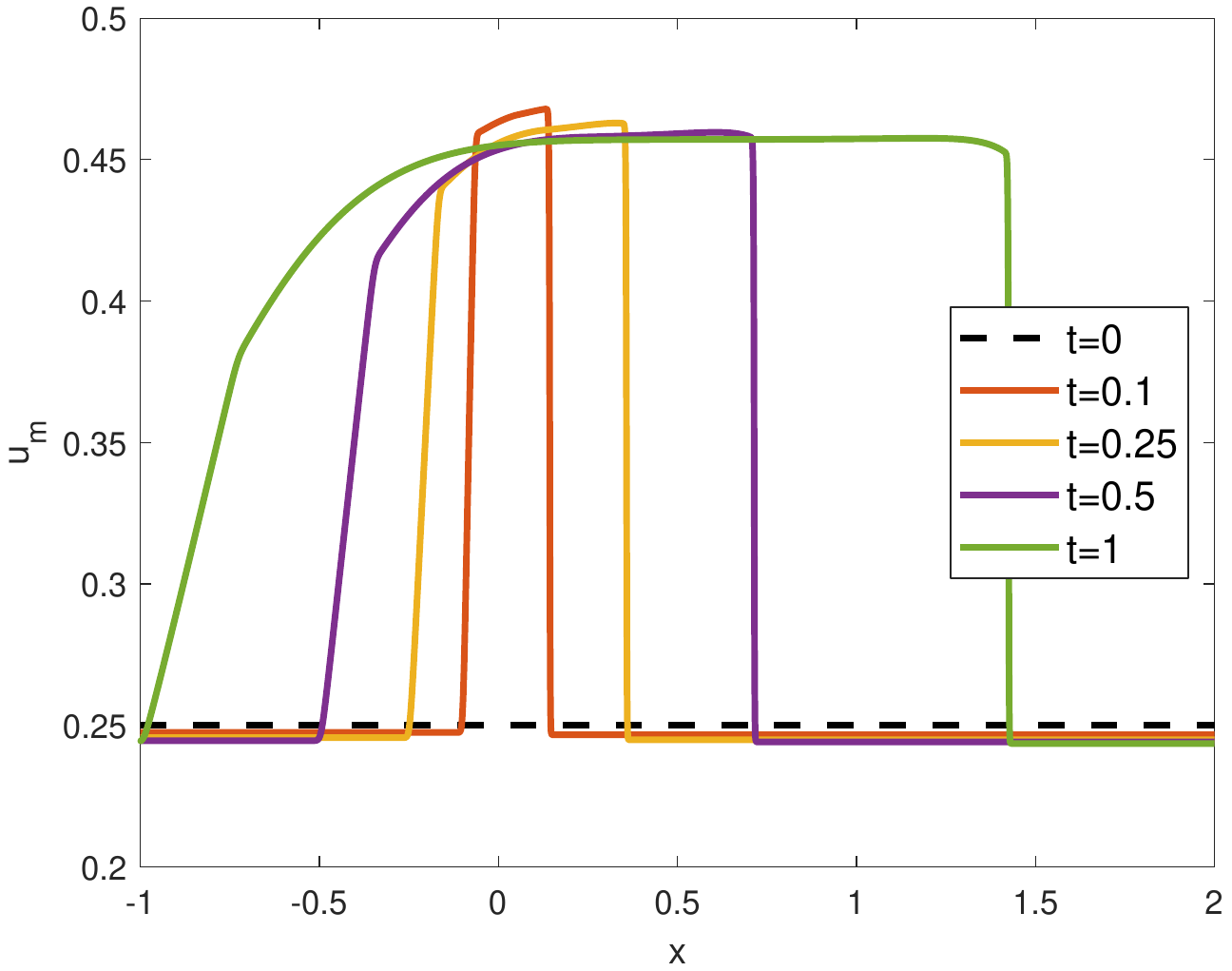}
			\caption{mean velocity $u_m$}
		    \label{fig:new_bottom-at-rest_u}
		\end{subfigure}
		\begin{subfigure}{0.32\textwidth}
		    \centering
			\includegraphics[width=0.99\textwidth]{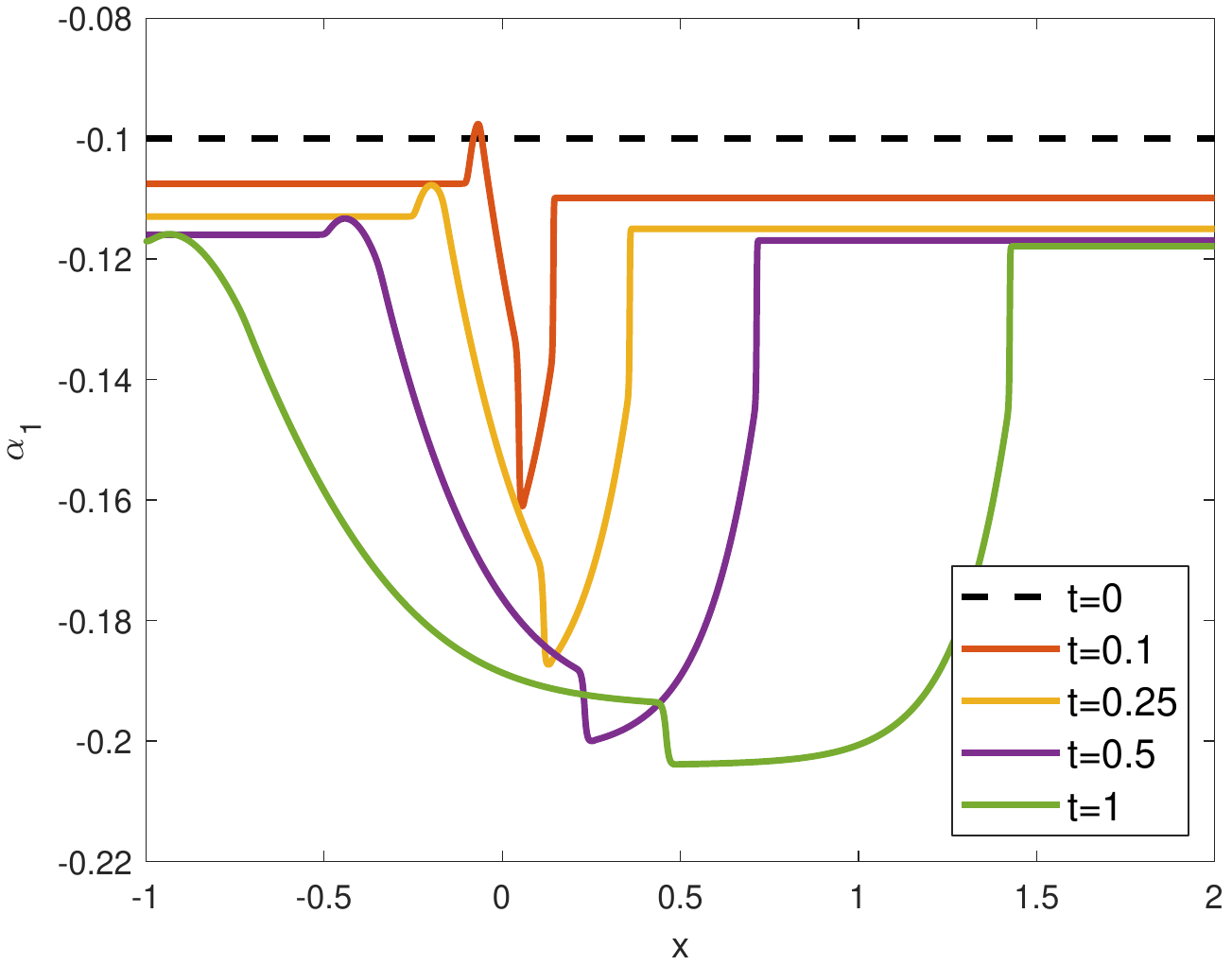}
			\caption{coefficient $\alpha_1$}
		    \label{fig:new_bottom-at-rest_alpha1}
		\end{subfigure}
		\center{
		\begin{subfigure}{0.32\textwidth}
		    \centering
			\includegraphics[width=0.99\textwidth]{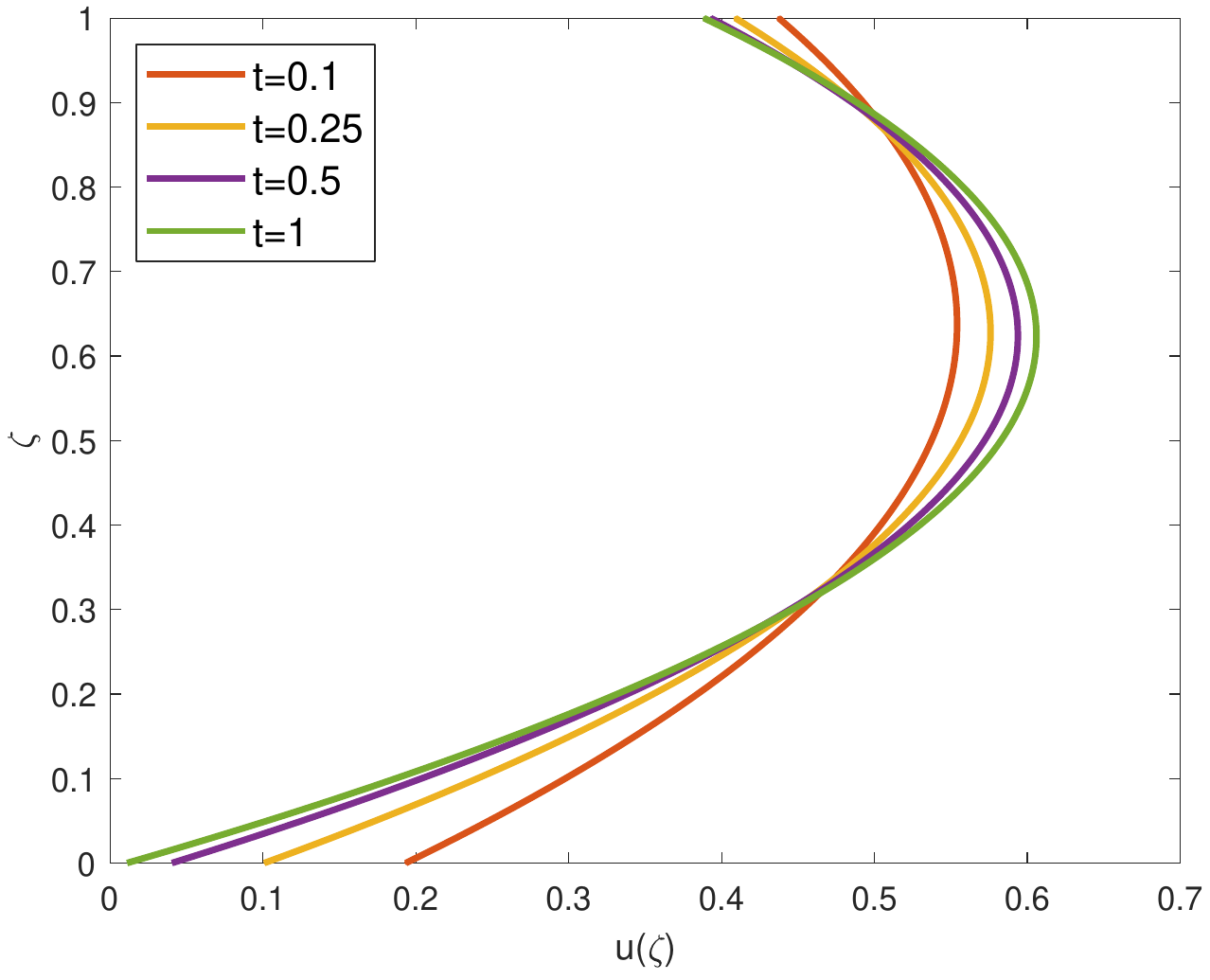}
			\caption{velocity profile $u(\zeta)$, $x=0$}
		    \label{fig:new_bottom-at-rest_uy}
		\end{subfigure}
		\begin{subfigure}{0.32\textwidth}
		    \centering
			\includegraphics[width=0.99\textwidth]{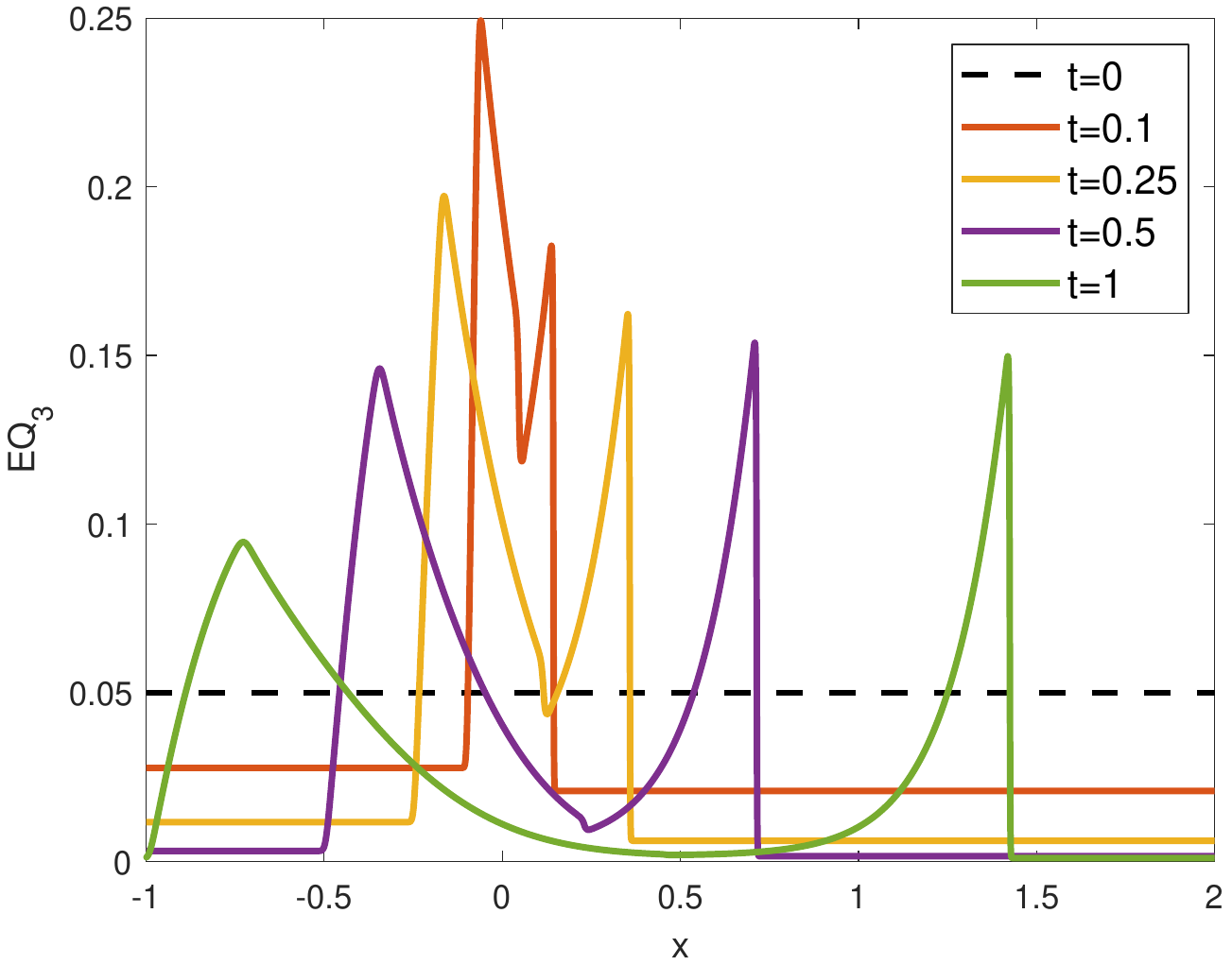}
			\caption{Deviation from equilibrium}
		    \label{fig:new_bottom-at-rest_EQ3}
		\end{subfigure}}
        \caption{For $\nu=10^{-3}, \lambda = 10^{-3}$, the SWLME model is converging to the bottom-at-rest equilibrium with time. No instability can be seen.}
        \label{fig:new_bottom-at-rest}
\end{figure}

As a summary of the numerical simulations it can be said that the expected convergence towards the three distinguished equilibrium manifolds could clearly be demonstrated. All equilibria exist as converged states of simulation scenarios. Despite the proven linear instability of the bottom-at-rest equilibrium, the SWLME model seems to be more stable than the HSWME and $\beta$-HSWME models. 

\section{Conclusions}
This paper performs a thorough analysis, both analytically and numerically, on the stability of several newly-developed shallow water moment models. To properly analyze the moment models, we first gave a concise but very general hyperbolicity proof for the HSWME and $\beta$-HSWME models. Next we identified three different equilibrium manifolds, which include the water-at-rest, the constant-velocity, and the bottom-at-rest profiles for the velocity. 
Each equilibrium manifold is attained by different limiting conditions for the friction coefficients. 
Our analysis revealed nonlinear stability of the water-at-rest equilibrium and the constant-velocity equilibrium. 
For the bottom-at-rest equilibrium, several counterexamples for linear instability were shown. However, the linear instability goes together with changing velocity sign, which can be interpreted as a violation of the shallow flow regime. We therefore infer that the equations are linearly stable as long as the flow conditions allow to treat the problem as shallow. In numerical tests, we obtained the same three equilibrium manifolds by choosing different friction parameters. A small non-linear instability could be seen for the bottom-at-rest equilibrium in the case of HSWME and $\beta$-HSWME, but the instability was not visible for the SWLME model due to its different structure. 
Building upon this stability analysis, further work should focus on numerical methods to properly preserve the equilibrium manifolds or guarantee positive velocity distributions. Additionally, high-order numerical schemes in the stiff regime can be investigated. 

\section*{Acknowledgements}
This research has been partially supported by the European Union's Horizon 2020 research and innovation program under the Marie Sklodowska-Curie grant agreement no. 888596.
Q. Huang is supported by the National Natural Science Foundation of China (Grant no. 51906122).
J. Koellermeier is a postdoctoral fellow in fundamental research of the Research Foundation - Flanders (FWO), funded by FWO grant no. 0880.212.840.


\bibliographystyle{plain}
\bibliography{references}
\end{document}